\documentclass[conference]{IEEEtran}
\usepackage{graphicx}
\usepackage{amsmath}
\usepackage{lipsum}
\usepackage{amsthm}
\usepackage[strings]{underscore}
\usepackage{amsfonts}
\usepackage{caption}
\usepackage{siunitx}
\usepackage{amssymb}
\usepackage{subcaption}
\usepackage{placeins}
\usepackage{lettrine}
\usepackage{fancyhdr} 
\usepackage{multicol, blindtext}
\usepackage{bbm}
\usepackage{calrsfs}
\usepackage{cite}
\usepackage{color}
\usepackage{cuted}
\usepackage{mathtools}

\usepackage{multicol}
\usepackage{hyperref}
\usepackage{algorithm}
\usepackage{algpseudocode}
\usepackage{authblk}
\newtheorem{remark}{Remark}
\makeatletter

\newcommand{\Rmnum}[1]{\expandafter\@slowromancap\romannumeral #1@}

\newtheorem{theorem}{Theorem}

\newtheorem{proposition}{Proposition}

\newtheorem{assumption}{Assumption}
\newtheorem{lemma}{Lemma}
\renewcommand\IEEEkeywordsname{Keywords}
\makeatletter
\def\BState{\State\hskip-\ALG@thistlm}
\makeatother
\ifCLASSINFOpdf
\else
\fi

\begin{document}
\title{A Framework for the Evaluation of Network Reliability Under Periodic Demand}
\author[*]{Ali Maatouk}
\author[*]{Fadhel Ayed}
\author[$\dagger$]{Shi Biao}
\author[*]{Wenjie Li}
\author[*]{Harvey Bao}
%\author[$\S$]{Dandan Miao}
%\author[$\S$]{Ke Lin}
%\author[$\dagger$]{Xin Chen}
\author[$\mathparagraph,\ddagger$]{Enrico Zio}
\affil[*]{Paris Research Center, Huawei Technologies, Boulogne-Billancourt, France}
\affil[$\dagger$]{École polytechnique, Palaiseau, France}
%\affil[$\S$]{Shanghai Research Center, Huawei Technologies, Shanghai, China}
\affil[$\mathparagraph$]{Centre de recherche sur les
Risques et les crises, MINES
Paris-PSL,
Paris, France}
\affil[$\ddagger$]{Energy Department, Politecnico di Milano, Milan, Italy}
\maketitle
\thispagestyle{fancy}
\pagestyle{fancy}
\fancyhf{}
\fancyheadoffset{0cm}
\renewcommand{\headrulewidth}{0pt} 
\renewcommand{\footrulewidth}{0pt}
\renewcommand*{\thepage}{\scriptsize{\arabic{page}}}
\fancyhead[R]{\thepage}
\fancypagestyle{plain}{%
   \fancyhf{}%
   \fancyhead[R]{\thepage}%
}

\begin{abstract} 
In this paper, we study network reliability in relation to a periodic time-dependent utility function that reflects the system's functional performance. When an anomaly occurs, the system incurs a loss of utility that depends on the anomaly's timing and duration. We analyze the long-term average utility loss by considering exponential anomalies' inter-arrival times and general distributions of maintenance duration. We show that the expected utility loss converges in probability to a simple form. We then extend our convergence results to more general distributions of anomalies' inter-arrival times and to particular families of non-periodic utility functions. To validate our results, we use data gathered from a cellular network consisting of 660 base stations and serving over 20k users. We demonstrate the quasi-periodic nature of users' traffic and the exponential distribution of the anomalies' inter-arrival times, allowing us to apply our results and provide reliability scores for the network. We also discuss the convergence speed of the long-term average utility loss, the interplay between the different network's parameters, and the impact of non-stationarity on our convergence results
\let\thefootnote\relax\footnotetext{This work has been presented in part at the 32nd European Safety and Reliability Conference (ESREL 2022) \cite{maatoukayyed}.}
\end{abstract}
\begin{IEEEkeywords}
Cellular networks, network reliability, expected utility not satisfied, periodic utility function.
\end{IEEEkeywords}
\IEEEpeerreviewmaketitle
\section{Introduction}
The evaluation of the reliability of any large-scale service system aims to assess its ability to provide services, taking into account the various hazardous events and anomalous conditions that can occur and impact the functioning of its components. These events, such as component failures, can affect the overall service reliability and availability of the system. In the literature on critical infrastructures, such as electric power grids, several measures are commonly used to evaluate reliability, including the loss of load probability (LOLP), expected frequency of load curtailment (EFLC), expected duration of load curtailment (EDLC), expected duration of curtailment (EDC), and expected demand not satisfied (EDNS) \cite{Al-Shaalan20,referencepower,bellani:hal-03481294,Medjoudj17}.

In recent years, the extent to which we rely on networked data systems as part of our critical national infrastructure has become increasingly apparent \cite{5762681}. As a result, the characterization of reliability metrics in networked systems, such as cellular networks, has become crucial. Traditionally, the literature on networked data systems has focused on maximizing link-level reliability to ensure the best possible service for connected devices (see references \cite{5759108}, \cite{6786441}, and \cite{9145854}). Most research on network-wide reliability has consisted of evaluating graph-based connectivity metrics in light of possible failures that may occur in the network (e.g., k-terminal network reliability, as described in reference \cite{9238406}). However, there has recently been a shift towards strict network-level reliability that takes into account user data traffic demand, particularly for next-generation cellular networks. For example, 5G networks are expected to provide 99.999\% (or "five nines") of data availability annually \cite{9700502}, with plans to improve to a seven-nines standard in 6G \cite{9349624}. Given the nature of these requirements, the characterization of the expected demand not satisfied metric in such systems becomes especially important \cite{Li2021}. However, evaluating this measure is challenging due to the complexity and randomness of anomalies, failures, outages, and maintenance in these systems.
\color{black} 

In this paper, we focus on quantitatively characterizing the expected utility not satisfied of a networked system. Specifically, we characterize reliability with respect to a time-dependent utility function $U(t)$ related to the system's functional performance. The utility function $U(t)$ can represent, for example, the data traffic typically transmitted at time $t$, the number of users served, or other similar quantities. Evaluating the system's reliability, then, involves deriving the long-term expected utility loss $\overline{\mathcal{L}}$ of the system, taking into account the failures that may occur at its components and external factors that influence the impact of such failures. Given the dependence of $U(t)$ on time, the evaluation of $\overline{\mathcal{L}}$ requires the formulation of the stochastic differential equations governing its evolution. Then, tools such as stochastic hybrid systems and Dynkin's formula are leveraged to analyze $\overline{\mathcal{L}}$ \cite{hespana,8469047,9007478}. However, obtaining a closed-form expression of $\overline{\mathcal{L}}$ is heavily contingent on the complexity of the differential equations involved. Generally, only approximations can be obtained by such analytical frameworks \cite{RePEc:eee:reensy:v:172:y:2018:i:c:p:159-170,2021arXiv210903919M}. Another approach to characterize $\overline{\mathcal{L}}$ consists of running Monte Carlo simulations of the system \cite{bouissou:hal-01182410,zio2013}. However, Monte Carlo simulations can be computationally costly, especially when a large number of components interact with one another to provide the system's functionality. Additionally, the absence of closed-form expressions reduces the interpretability of $\overline{\mathcal{L}}$ and hinders the optimization process of the system's parameters in the design stage. Hence, the goal of our paper is to address these challenges and provide a theoretical framework to obtain an expression of $\overline{\mathcal{L}}$ under a periodicity assumption on $U(t)$. To that end, the following are the key contributions of this paper:
%(coinciding with the EDNS metric)
\begin{itemize}
\item We begin our stochastic analysis by formulating the expected utility not satisfied $\overline{\mathcal{L}}$ as a function of various elements, such as the inter-arrival times of anomalies and their repair times. Through Fourier analysis of the anomalies' inter-arrival time distribution, we show that a key stochastic process converges to a uniform distribution. We then leverage this convergence to provide limiting distributions for several quantities that impact the system's performance.
\item Next, we combine our earlier results and the periodicity of the utility function to show that the expected utility not satisfied converges in probability to a simple and intuitive form. We also establish connections between this expression and standard availability metrics.
\item Afterward, we extend our theoretical results to the case where the periodicity assumption is slightly violated. Specifically, we show that our convergence results still hold when the periodic utility is corrupted by a random stochastic process that satisfies predefined conditions on its first and second-order statistics. Additionally, through Fourier analysis, we also demonstrate that our assumption of exponential distribution for the anomalies' inter-arrival times can be relaxed, and our theoretical findings hold for any bounded probability density function of the anomalies' inter-arrival times.
\item Finally, we consider a large-scale cellular network with 660 cells serving over 20,000 users. Using data we obtained from the network operator, we demonstrate the quasi-periodic nature of user traffic and the exponential distribution of anomalies' inter-arrival times. We, then, use our theoretical results to characterize the expected data traffic not satisfied for this specific network. Additionally, we show that the expected utility not satisfied converges quickly to its theoretical limit, and we investigate the impact of non-stationarity on the convergence results.
\end{itemize}
The rest of the paper is organized as follows: Section \ref{systmodel}
introduces the system model adopted in the paper. In Section \ref{mathanalysis}, we present the mathematical analysis of our system, and we provide our main theoretical findings. Afterward, we extend our theoretical analysis to more general settings in Section \ref{generalizationss}. Then, a use case scenario consisting of a large-scale cellular network is considered in Section \ref{usecasescellular}, and our theoretical findings are then corroborated. Lastly, Section \ref{conclusionsss} concludes the paper. 
\section{System Model}
\label{systmodel}
We consider a system operating in its useful-life phase, during which independent anomalies occur randomly with a Poisson rate $\lambda$ \cite[Chapter~3]{doi:https://doi.org/10.1002/9781118841716.ch3},\cite{OHRING1995747}. Thus, the inter-arrival time $X_j$ between anomalies $j-1$ and $j$ is exponentially distributed with rate $\lambda$
\begin{equation}
\Pr(X_j<t)=1-e^{-\lambda t}.
\end{equation}
When an anomaly takes place, the system operator triggers a maintenance procedure. We let $Y_j$ denote the maintenance time associated with anomaly $j$. We make the assumption that the variables $Y_j$ are independent of $X_j$ and are independently and identically distributed (i.i.d.) with a cumulative distribution function described by
\begin{equation}
\Pr(Y_j<t)=\begin{cases}
 F_{Y}(t), & \text{for } t\geq0, \\
    0, & \text{otherwise}.
  \end{cases}
\end{equation} 
This means that for each $j=1, \ldots$, the variable $Y_j$ follows the same distribution. We also make the following assumption on the distribution of $Y_j$.
\begin{assumption}
The first and second order moments $\mathbb{E}[Y_j]$ and $\mathbb{E}[Y^2_j]$ are finite for $j\in\mathbb{N}^*$.
\label{assumptionony}
\end{assumption}
Furthermore, we consider that when an anomaly takes place at time $t_0$ and the problem is resolved at time $t_1$, a utility loss $\int_{t_0}^{t_1}U(t)dt$ is incurred. Letting $W(t)$ be a binary random variable that is equal to $1$ when the system is suffering from an anomaly and $0$ otherwise, we can define the expected utility not satisfied as follows
\begin{equation}
\overline{\mathcal{L}}=\lim_{T\to+\infty} \:\frac{1}{T}\int_{0}^{T}U(t)W(t)dt.
\label{expectedlossfirstform}
\end{equation}
In practice, the consideration of large $T$ amounts to the system being operated long enough before the expected utility loss assessment. 
%A schematic, fictitious illustration of the evolution in time of the utility loss can be found in Fig. \ref{sourceexample}. 
%\begin{figure}[!ht]
%\centering
%\includegraphics[width=.99\linewidth]{utilityloss}
%\caption{Loss function evolution in time}
%\label{sourceexample}
%\end{figure}\\
Note that the function $U(t)$ can represent a large variety of system quantities depending on the system's operator priorities. For example:
\begin{itemize}
\item $U(t)$ can denote the customer demand (e.g., electricity, communication traffic) served by the system at time $t$; in this scenario, the expected loss score coincides with the notion of EDNS \cite{Medjoudj17}. 
\item $U(t)$ can represent the number of users served by the system at time $t$; thus, the expected utility loss in this case represents the expected number of users affected by the anomalies.
\item $U(t)=\beta(t)U'(t)$, where $U'(t)$ denotes the utility of the system at time $t$ (e.g., customer demand) and $\beta(t)$ is a factor between $0$ and $1$ that indicates the fraction of utility lost depending on the occurrence time of the anomaly.  
\end{itemize}
One can also define $U(t)$ as a combination of system quantities. This shows the generality of our framework. In the next section, we will illustrate the mathematical framework to characterize $\overline{\mathcal{L}}$. But, first, let us consider the following assumption for $U(t)$.
\begin{assumption}
The utility function is a non-negative bounded periodic function. Precisely, there exists a constant $K>0$ and a period $p>0$ such that
\begin{align}
0\leq U(t)&\leq K,\nonumber\\
U(t+p)&=U(t),
\end{align}
for $t\geq0$. 
\label{assumptiononut}
\end{assumption}
The interest in the above assumption is that, as we will show in later sections,  when $U(t)$ verifies the above assumption, $\overline{\mathcal{L}}$ ends up converging in probability to a relatively simple form, which relaxes the difficulties of the quantitative characterization of the reliability of complex networked systems. Note that the periodicity of system utility has been observed in various practical applications due to the nature of human behavior with respect to service demand. For example, in cellular networks, it was shown that user traffic typically exhibits a periodical pattern \cite{7762185}. Such patterns have been also found in data that we gathered from a large-scale cellular network, as will be illustrated in Section \ref{subsectionutility}. It is worth noting that such trends are not exclusive to cellular networks. For instance, the periodic behavior has also been witnessed in the electricity demands in power grid networks \cite{YUKSELTAN2017287}. Accordingly, our theoretical derivations are not constrained to cellular networks settings, but rather can be leveraged for various other application scenarios.

%\begin{figure}[htb]
%\centering
%\includegraphics[width=.99\linewidth]{dlavg}
%\includegraphics[width=.99\linewidth]{rrcconnavg}
%\caption{LTE traffic demand(top) and number of users connected to the base station(bottom)}
%\label{dlavg}
%\end{figure}

%Computing this measure is, however, challenging due to the complexity of the network systems and the randomness of the occurrence of anomalies failures, and outage events, as well as the maintenance duration. In this paper, we focus on quantitatively characterizing the expected utility not satisfied as a reliability measure of a network system. Specifically, reliability is characterized with respect to a time-dependent utility function $U(t)$ related to the system functional performance. The utility function $U(t)$ can represent, for example, the data traffic served by a cell in cellular networks (coinciding with the EDNS framework) or the number of users served, or other quantities alike. The reliability evaluation, then, amounts to deriving the expected utility loss $\overline{\mathcal{L}}$ of the system considering the failures that may occur at its components. And the external factors, like function demand, which influence the impact of such failures. For example, in a cellular network, an anomaly occurring at 2 AM when user activity is low, would lead to a lower utility loss than an anomaly taking place at peak users' activity hours. 

\section{Mathematical Analysis}
\label{mathanalysis}
\begin{figure*}[ht]
\centering
\begin{subfigure}{0.33\textwidth}
  \centering
  \includegraphics[width=.99\linewidth]{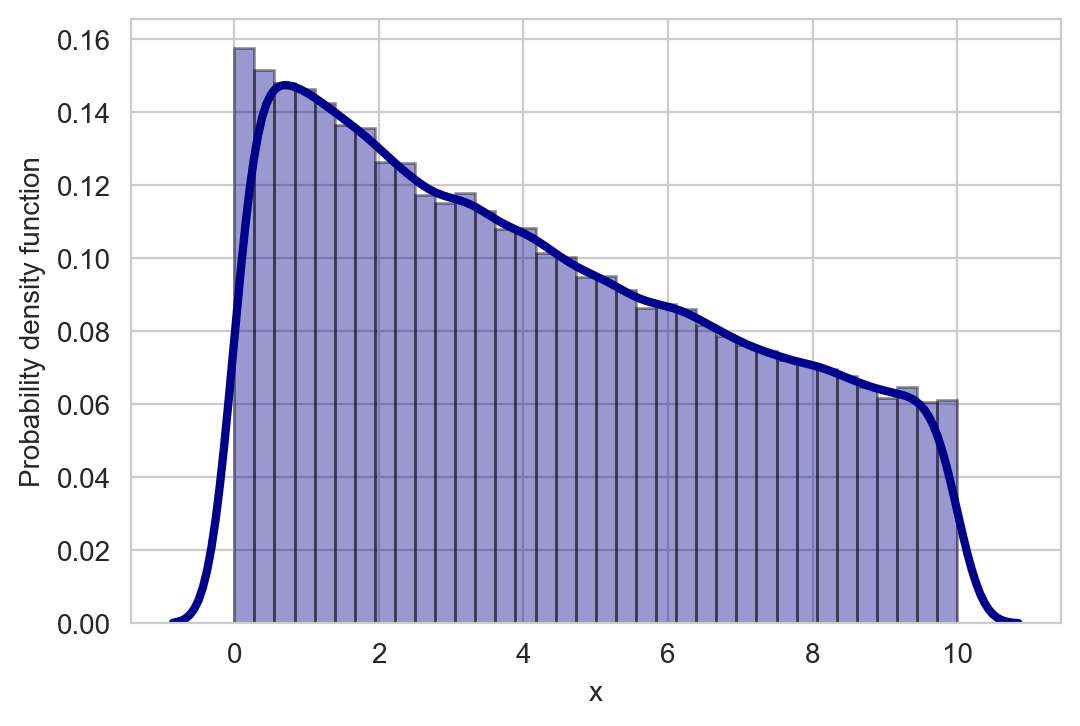}
  \caption{Zero Gaussian added.}
    \label{agemetric}
\end{subfigure}%    
\begin{subfigure}{0.33\textwidth}
\centering
  \includegraphics[width=.99\linewidth]{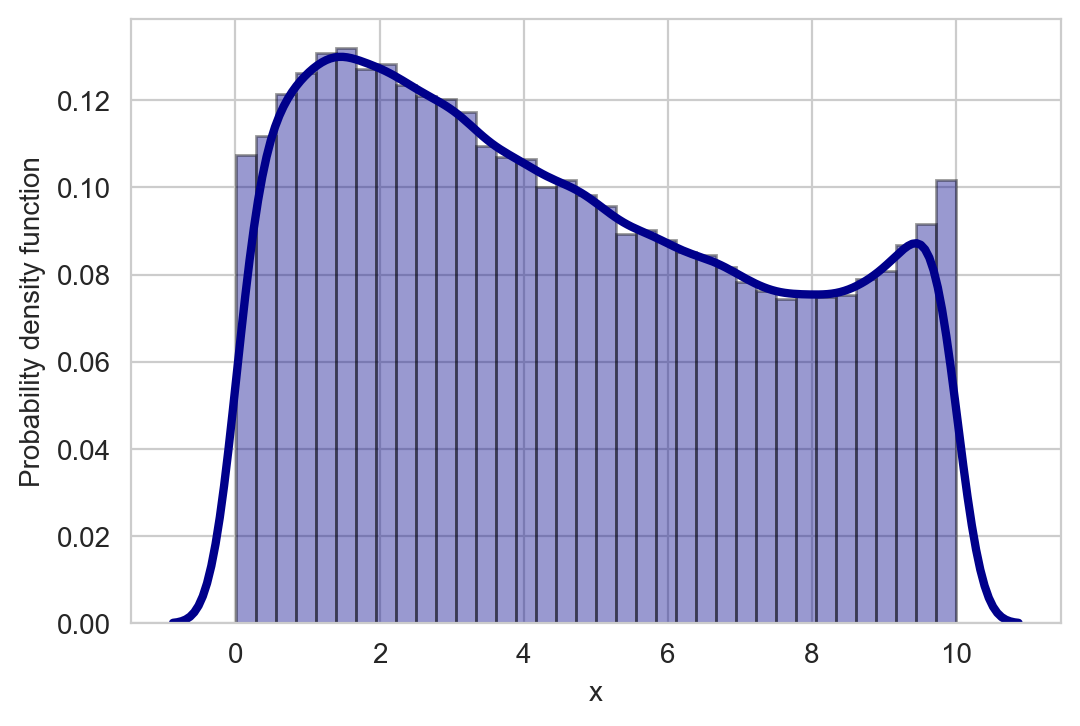}
  \caption{One Gaussian added.}
\label{errmetric}
\end{subfigure}%
\begin{subfigure}{0.33\textwidth}
\centering
  \includegraphics[width=.99\linewidth]{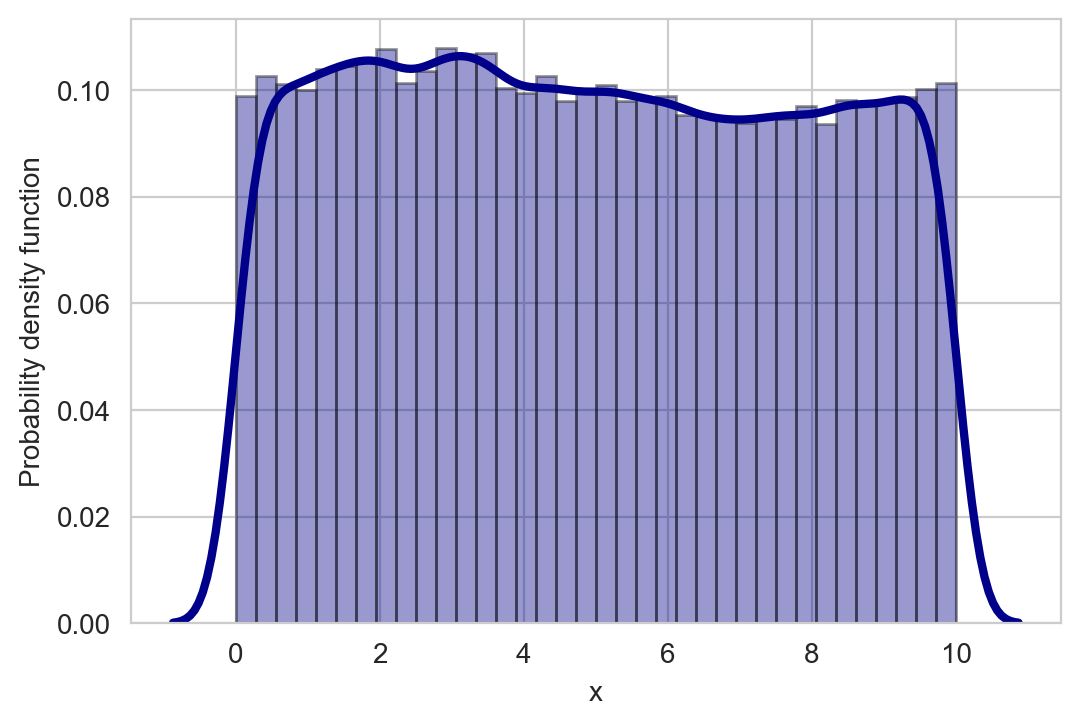}
  \caption{Ten Gaussians added.}
\label{effmetric}
\end{subfigure}%
\caption{Illustrations of the smoothing phenomenon.}
\vspace{-13pt}
\label{metrics}
\end{figure*}
We decompose the time horizon $T$ in eq. (\ref{expectedlossfirstform}) into multiple stages. Specifically, we let $D_n=\sum_{j=1}^{n}(X_j+Y_j)$ 
and we rewrite the expected utility loss of the system as
\begin{equation}
\overline{\mathcal{L}}=\lim_{n\to+\infty} \:\frac{\int_{0}^{D_n}U(t)W(t)dt}{D_n}.
\end{equation}
Next, by multiplying by $\frac{1}{n}$ both the numerator and denominator, we end up with
\begin{equation}
\overline{\mathcal{L}}=\lim_{n\to+\infty} \:\frac{\frac{1}{n}\sum_{j=1}^{n}\int_{D_{j-1}}^{D_j}U(t)W(t)dt}{\frac{1}{n}\sum_{j=1}^{n}(X_j+Y_j)}.
\end{equation}
Noting that $W(t)$ is equal to $0$ by definition in every interval $[D_j,D_j+X_j]$, we can rewrite the expected loss as
\begin{equation}
\overline{\mathcal{L}}=\lim_{n\to+\infty} \:\frac{\frac{1}{n}\sum_{j=1}^{n}\int_{D_j-Y_j}^{D_j}U(t)dt}{\frac{1}{n}\sum_{j=1}^{n}(X_j+Y_j)}.
\label{expectedlossbeforefinal}
\end{equation}
Clearly, the challenging part of the evaluation of the expected loss is the numerator. To deal with this, we leverage the periodicity of the function $U(t)$. 
\begin{lemma}
If $U(t)$ is a periodic function of period $p$, then the expected loss $\overline{\mathcal{L}}$ can be rewritten as
\begin{equation}
\overline{\mathcal{L}}=\lim_{n\to+\infty} \:\frac{\frac{1}{n}\sum_{j=1}^{n}\int_{D_j ^{[p]}-Y_j}^{D^{[p]}_j}U(t)dt}{\frac{1}{n}\sum_{j=1}^{n}(X_j+Y_j)},
\label{expressiontaba3lasess}
\end{equation}
where $D_j^{[p]}=D_j \text{ mod }p$ is the remainder of the Euclidean division of $D_j$ by $p$ (i.e., the least positive residue).
\end{lemma}
\begin{proof}
To prove this result, we first note that
\begin{equation}
D_j=kp+(D_j \text{ mod }p),
\end{equation}
where $k\in\mathbb{N}$. Next, we apply a change of variable to the integral in eq. (\ref{expectedlossbeforefinal}), letting $t'=t-kp$. By doing so, eq. (\ref{expectedlossbeforefinal}) can be rewritten as follows
\begin{align}
\overline{\mathcal{L}}&=\lim_{n\to+\infty} \:\frac{\frac{1}{n}\sum_{j=1}^{n}\int_{D_j^{[p]}-Y_j}^{D_j^{[p]}}U(t'+kp)dt'}{\frac{1}{n}\sum_{j=1}^{n}(X_j+Y_j)}\nonumber\\&\overset{(a)}{=}\lim_{n\to+\infty} \:\frac{\frac{1}{n}\sum_{j=1}^{n}\int_{D_j^{[p]}-Y_j}^{D_j^{[p]}}U(t')dt'}{\frac{1}{n}\sum_{j=1}^{n}(X_j+Y_j)},
\end{align}
where $(a)$ results from the periodicity of $U(t)$. Then, by interchanging $t'$ and $t$, we can confirm the lemma.
\end{proof}
The next step in the analysis consists of finding the distribution of $D_j^{[p]}$. To do so, we first rewrite  $D_j^{[p]}$ as follows
\begin{equation}
D_j^{[p]}=(X^{[p]}_j+Y_j^{[p]})\: \text{mod}\: p,
\label{rewritedj}
\end{equation}
where 
\begin{align}
&X^{[p]}_j=(\sum_{k=1}^{j}X_k)\: \text{mod}\: p,\nonumber\\&Y^{[p]}_j=(\sum_{k=1}^{j}Y_k)\: \text{mod}\: p.
\label{modulxpstuff}
\end{align}
Next, we investigate the distribution of $X^{[p]}_j$ more closely. 
\begin{theorem}
Let $f_{X^{[p]}_j}(x)$ denote the probability distribution function of $X^{[p]}_j$. We have that $f_{X^{[p]}_j}(x)$ converges uniformly to the uniform distribution on $[0,p]$. Precisely, 
\begin{equation}
\sup_{x\in[0,p]} |f_{X^{[p]}_j}(x)-\frac{1}{p}|\leq \frac{C\alpha^{j}}{p},
\end{equation}
where
\begin{equation}
\alpha=\frac{\lambda p}{\sqrt{\lambda^2p^2+4\pi^2}},
\end{equation}
\begin{equation}
C=\sum_{n=1}^{\infty}\frac{\lambda^2p^2+4\pi^2}{\lambda^2p^2+4\pi^2n^2}.
\end{equation}
\label{theoremuniform}
\end{theorem}
\begin{proof}
The proof revolves around a Fourier analysis of the distribution of $X^{[p]}_j$. In essence, we first analyze the behavior of the Fourier coefficients of the random variables making up $X^{[p]}_j$. Then, by leveraging the scalability property of the exponential distribution, we can derive the desired results. For conciseness, the details of the proof are reported in Appendix \ref{prooffourrier}.
\end{proof}
In this theorem, we showed that $X^{[p]}_j$ converges uniformly to a uniform distribution as $j$ becomes large. Additionally, the bounds we derived imply that $\nabla\alpha(\lambda,p)\geq0$ for $\lambda,p\geq0$. This means that the speed of convergence to a uniform distribution decreases as either $p$ or $\lambda$ increases. Figure \ref{alphainfunctionplambda} illustrates this trend, where $\alpha$ varies and its contour levels are shown. To understand this trend, we can examine how the cumulative density function of the anomalies' inter-arrival times changes with increasing $\lambda$. For any interval $\mathcal{I}=[0,a]$ where $a$ is any positive real number, $\Pr(X_k\in\mathcal{I})=1-\exp(-\lambda a)$ increases as $\lambda$ grows for $k\in\mathbb{N}^*$. Based on this and the expression for $X^{[p]}_j$ given by eq. (\ref{modulxpstuff}), we can conclude that more realizations of $X_k$ are needed to span all the possible values from $[0,p]$. Furthermore, as $p$ increases, the range of values where $X^{[p]}_j$ can fall also increases. Hence, using the expression for $X^{[p]}_j$, we can see that more realizations of $X_k$ are also needed in this case to span all the possible values of $X^{[p]}_j$ from $[0,p]$. In both cases, larger values of $j$ are needed in order to converge.
\begin{figure}[!ht]
\centering
\includegraphics[width=.99\linewidth]{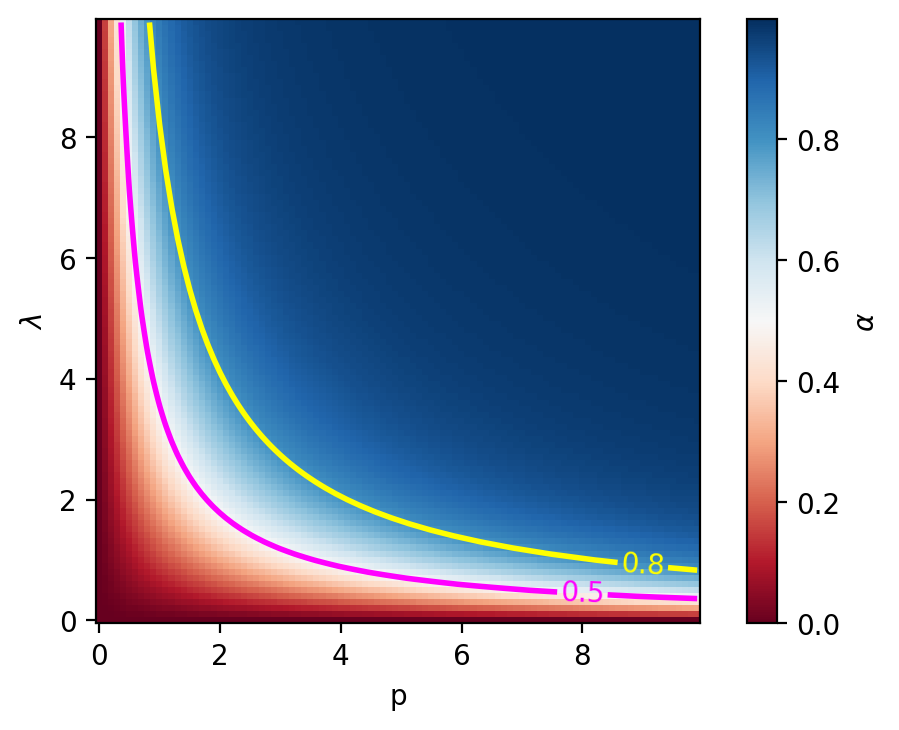}
\vspace{-20pt}
\caption{Illustration of $\alpha$ in function of $\lambda$ and $p$.}
\label{alphainfunctionplambda}
\end{figure}\\
Next, we recall that our goal is to characterize the distribution of $D^{[p]}_j$. To do so, we need to take into account the distribution of $Y^{[p]}_j$, which can be quite general as we impose no restriction on $F_{Y}(t)$. To alleviate this difficulty, we provide below an essential lemma to our analysis. 
\begin{lemma}
Let $A$ be a random variable defined on $[0,p]$ satisfying
\begin{equation}
\sup_{x\in[0,p]}|f_{A}(x)-\frac{1}{p}|\leq C,
\end{equation}
where $f_{A}(x)$ is the probability density function of $A$. Let $B$ be a random variable of arbitrary distribution defined on $[0,p]$ independent of $A$. Then,
\begin{equation}
\sup_{x\in[0,p]} |f_{Z}(z)-\frac{1}{p}|\leq \sup_{x\in[0,p]}|f_{A}(x)-\frac{1}{p}|\leq C,
\end{equation}
where $Z=A+B\text{ mod }p$.
\label{addingsmoothesout}
\end{lemma}
\begin{proof}
Given the independence between $A$ and $B$, the proof revolves around the notion of probability distributions' convolution. Then, by leveraging the particularity of the distribution of $A$ along with the definition of the modulo function, we can derive the desired results. For conciseness, the details of the proof are reported in Appendix \ref{proofsum}.
\end{proof}
\begin{remark}
The above lemma has an important interpretation: when a random variable (RV) $B$ independent of $A$ is added to the random variable $A$, it results in a smoothing effect when the modulo $p$ function is applied. Thus, the distribution of $Z=A+B\text{ mod }p$ becomes even closer to a uniform distribution than $A\text{ mod }p$. To illustrate this, we consider an exponentially distributed RV $A$ of rate $\lambda=10$, and we add to this RV $0$, $1$, and $10$ independent standard Gaussian RVs. The results are reported in Fig. \ref{metrics}, where the distribution of the sum of the RVs modulo $p$ is plotted in each case, and the smoothing effect is shown. 
\end{remark}
With the above lemma in mind, we can now tackle the characterization of the distribution of $D^{[p]}_j$ in the next proposition. 
\begin{proposition}
Let $f_{D^{[p]}_j}(z)$ denote the probability distribution function of $D^{[p]}_j$. We have that $f_{D^{[p]}_j}(z)$ converges uniformly to the uniform distribution on $[0,p]$. Precisely, 
\begin{equation}
\sup_{z\in[0,p]} |f_{D^{[p]}_j}(z)-\frac{1}{p}|\leq \frac{C\alpha^{j}}{p}.
\end{equation}
%Let $A_j$ be a RV that converges uniformly to a uniform distribution on $[0,p]$ and $B$ a RV of arbitrary distribution defined on $[0,p]$, independent of $A_j$. Then, the sum $A_j+B \: \text{mod}\: p$ also converges uniformly to $[0,p]$ with the same convergence rate of that of $A$. 
\label{theoremadduniform}
\end{proposition}
\begin{proof}
To prove the proposition, we recall from eq. (\ref{rewritedj}) that $D^{[p]}_j$ can be written as the sum of two RVs defined on $[0,p]$. Then, due to the independence between $X^{[p]}_j$ and $Y_j^{[p]}$, we can leverage Lemma \ref{addingsmoothesout} to conclude the proposition.
\end{proof}
%Given the above, we can conclude that the RV 
%\begin{equation}
%D^{[p]}=(X^{[p]}+Y^{[p]})\text{ mod }p
%\end{equation}
%is uniformly distributed on $[0,p]$ where 
%\begin{align}
%&D^{[p]}=\lim_{j\to+\infty}D^{[p]}_j,\nonumber\\&Y^{[p]}=\lim_{j\to+\infty}Y^{[p]}_j.
%\end{align}
\begin{remark}
The trend seen in Fig. \ref{metrics} is confirmed by Proposition \ref{theoremadduniform}, which shows that the convergence of the distribution $D^{[p]}_j$ to a uniform distribution happens exponentially fast. This suggests that a relatively small number of anomalies/repair stages is sufficient to model $D^{[p]}_j$ as a uniform distribution with minimal modeling penalty. As seen in Fig. \ref{metrics}, even when only 10 independent Gaussian random variables are added to the exponential random variable, the resulting sum modulo $p$ is already distributed almost uniformly.
\end{remark}
Given the above results, we can conclude that as $j$ gets large, the distribution of $D^{[p]}_j$ approaches the uniform distribution on $[0,p]$. Now, to characterize $\overline{\mathcal{L}}$ in eq. (\ref{expressiontaba3lasess}), let us define the integral $I_j$ as 
\begin{equation}
I_j=\int_{D_j ^{[p]}-Y_j}^{D^{[p]}_j}U(t)dt.
\end{equation}
As seen in eq. (\ref{expressiontaba3lasess}), the difficult part in characterizing $\overline{\mathcal{L}}$ comes from the fact that the terms $I_j$ in the numerator are not independent and do not share the same distribution. Hence, providing a statistical convergence of the sum present in the numerator is challenging. To address these challenges, we first provide the following lemma for the convergence of the sum of the expected values of $I_j$. 
\begin{proposition}
Let $S_n=\frac{1}{n}\sum_{j=1}^{n}\mathbb{E}[I_j]$ denote the partial sum of the expected value of $I_j$. We have 
\begin{equation}
S_n\xrightarrow{n \to +\infty}\overline{I},
\end{equation}
where $\overline{I}=\mathbb{E}[Y]\overline{U}$, and $\overline{U}=\frac{1}{p}\int_{0}^{p}U(t)dt$ is the average utility in a period $p$.
\label{convergenceofthesum}
\end{proposition}
\begin{proof}
Our proof consists of leveraging Lemma \ref{addingsmoothesout} and Proposition \ref{theoremadduniform}, the periodicity of $U(t)$, and the Stolz–Cesàro theorem. 
The details of the proof are reported in Appendix \ref{proofcesaro}.
%
%
%
%Let us now investigate the difference between $\overline{I}$ and $\lim_{n\to+\infty}\mathbb{E}[I_{n+1}]$. To that end, we rewrite the difference as follows
%\begin{align}
%|\lim_{n\to+\infty}\mathbb{E}[I_{n+1}]-\overline{I}|&\overset{a}=|\lim_{n\to+\infty}
%\int_{0}^{\infty}\int_{0}^{p}\int_{v-y}^{v}
%\mathbb{E}_{Y,V}[\int_{V-Y}^{V}U(t)dt|,
%\end{align}
%
%eq. (\ref{expectedlossbeforefinal})
%
%
%
%
%
%
%
%
%
%
%
%
%\begin{equation}
%\lim_{n\to+\infty}\mathbb{E}[I_{n+1}]=\lim_{n\to+\infty}\int_{0}^{\infty}\int_{0}^{p}\int_{v-y}^{v}U(t)dtf_{Y}(y)dy
%\end{equation}
%
%
%
%
%
%
%
%where $f_{Y}(y)$ is the probability density function of $Y_j$.
%
%
%
%
%
%
%Next, we develop the expression above as follows
%\begin{equation}
%\overline{I}=\mathbb{E}_{Y}\big[\frac{1}{p}\int_{0}^{p}\int_{v-Y}^{v}U(t)dtdv\big].
%\end{equation}
%To further simply the expression above, we proceed with a change of variable. Precisely, we let $s=t-v$. Consequently, we end up with
%\begin{align}
%\overline{I}&=\mathbb{E}_{Y}\big[\frac{1}{p}\int_{0}^{p}\int_{-Y}^{0}U(s+v)dsdv\big]\nonumber\\&\overset{(a)}{=}\mathbb{E}_{Y}\big[\frac{1}{p}\int_{-Y}^{0}\int_{0}^{p}U(s+v)dvds\big]\nonumber\\&\overset{(b)}{=}\mathbb{E}_{Y}\big[\frac{1}{p}\int_{-Y}^{0}\int_{0}^{p}U(v)dvds\big]\nonumber\\&=\mathbb{E}[Y]\frac{1}{p}\int_{0}^{p}U(v)dv,
%\end{align}
%where $(a)$ is the result of changing the order of integration, and $(b)$ is due to the periodicity of $U(t)$. Then, by interchanging $v$ and $t$, we can conclude the simplified form of $\overline{I}$. 

\end{proof}
Given the above convergence results, the next step of the analysis is to investigate the degree of dependence between each of the integrals $I_j$ and $I_{j+k}$ for any $j,k\in\mathbb{N}^*$. To do so, we investigate the joint distribution function of $D^{[p]}_{j}$ and $D^{[p]}_{j+k}$, as seen below. 
\begin{proposition}
Let $f_{D^{[p]}_{j},D^{[p]}_{j+k}}(z_j,z_{j+k})$ denote the joint distribution function of the RVs $D^{[p]}_{j}$ and $D^{[p]}_{j+k}$. There exists a constant $C'>0$ such that
\begin{equation}
|f_{D^{[p]}_{j},D^{[p]}_{j+k}}(z_j,z_{j+k})-\frac{1}{p^2}|\leq C'(\frac{C\alpha^{j}}{p}+\frac{C\alpha^{k}}{p}),
\end{equation}
for all $z_j,z_{j+k}\in[0,p]$ and $j,k\in\mathbb{N}^*$.  
\label{jointresults}
\end{proposition}
\begin{proof}
In this proof, we derive the joint probability distribution and leverage the particularities of the modulo function along with the independence between the RVs $X_i$ and $Y_i$ involved. For conciseness, the proof is reported in Appendix \ref{appendixjoint}.
\end{proof}
Now that the convergence of $S_n$ is proven and the joint probability distribution $f_{D^{[p]}_{j},D^{[p]}_{j+k}}(\cdot,\cdot)$ is characterized, we can investigate the covariance of $I_j$ and $I_{j+k}$. This will allow us to measure the dependency between each of these integrals, thus enabling us to derive $\overline{\mathcal{L}}$. 
\begin{proposition}
The covariance $\mathrm{Cov}[I_j,I_{j+k}]$ for $j,k\in\mathbb{N}^*$ satisfies the following inequality
\begin{align}
\mathrm{Cov}[I_j,I_{j+k}]\leq & \overline{I}^2(1+pC'C(\alpha^{j}+\alpha^{k})-(1-\frac{C\alpha^{j-1}}{p})\nonumber\\&(1-\frac{C\alpha^{j+k-1}}{p})).
\end{align}
\label{covarianceproposition}
\end{proposition}
\begin{proof}
The proof revolves around a derivation of the covariance of $I_j$ and $I_{j+k}$, and leveraging the propositions and lemmas provided thus far, along with various algebraic manipulations. The details of the proof are reported in Appendix \ref{appendixcovarianceproposition}. 
\end{proof}
Interestingly, we can conclude from the results above that when $j$ is large, and the gap between $j$ and $j'>j$ is also large, then the two integrals $I_j$ and $I_{j'}$ become uncorrelated. This is going to be essential to establish the main results of our mathematical analysis below.
\begin{theorem}
\label{theoremfinal}
If (i) the anomalies inter-arrival times are exponentially distributed, (ii) the maintenance times verify Assumption \ref{assumptionony}, and (iii) $U(t)$ satisfies Assumption \ref{assumptiononut}, then the expected loss $\overline{\mathcal{L}}$ converges in probability to 
\begin{equation}
\overline{\mathcal{L}}_{\infty}=\frac{\mathbb{E}[Y]\overline{U}}{\mathbb{E}[X]+\mathbb{E}[Y]}.
\end{equation}%$$. Precisely,  
%\begin{equation}
%\Pr(|\overline{\mathcal{L}}-\frac{\mathbb{E}[Y]\overline{U}}{\frac{1}{\lambda}+\mathbb{E}[Y]}\geq0\}
%\end{equation}
%can be rewritten as
%\begin{equation}
%\overline{\mathcal{L}}=\frac{\mathbb{E}[Y]\overline{U}}{\frac{1}{\lambda}+\mathbb{E}[Y]}\xrightarrow{T \to +\infty}0
%\end{equation}
\end{theorem}
\begin{proof}
The proof of the theorem amounts to characterizing the probability that the difference between the numerator and $\overline{I}$ exceeds a certain arbitrarily small $\epsilon$. Using the triangular inequality and the union bound, this probability is further simplified. Then, by leveraging Chebyshev's inequality and the results presented in the paper thus far, we obtain the desired characterization. Finally, using the weak law of large numbers and the Mann-Wald theorem, we can conclude the convergence in probability of $\overline{\mathcal{L}}$ to $\overline{\mathcal{L}}_{\infty}$. For conciseness, the details of the proof are reported in Appendix \ref{appendixtheoremfinal}. 
\end{proof}
As can be seen above, the expression to which $\overline{\mathcal{L}}$ converges turns out to have a relatively simple formulation thanks to the periodicity of the function $U(t)$. Particularly, given the periodic nature of $U(t)$, the challenges in evaluating $\overline{\mathcal{L}}$ due to the dynamic nature of the utility function are overcome, hence allowing us to evaluate  $\overline{\mathcal{L}}$ in a simplified manner. Additionally, the expression has an intuitive meaning and relationship with well-established metrics. In fact, considering the system's availability expression
\begin{equation}
\textnormal{Availability}=\dfrac{\textnormal{MTBF}}{\textnormal{MTBF}+\textnormal{MTTR}},
\end{equation}
where MTBF$=\mathbb{E}[X]=\frac{1}{\lambda}$ and MTTR$=\mathbb{E}[Y]$ denote the mean time between failure and mean time to repair, respectively \cite{Medjoudj17}, and by examining the expression of $\overline{\mathcal{L}}_{\infty}$, we can deduce that
\begin{equation}
\overline{\mathcal{L}}_{\infty}=(\underbrace{1-\textnormal{Availability}}_\text{(A)})\times\underbrace{\overline{U}}_\text{(B)}.
\end{equation}
The term (A) can be seen as the probability that the system suffers from an anomaly; the term (B) is the average utility that the system delivers in a period $p$. Then, one can see that the periodicity of the utility $U(t)$, the exponential nature of anomalies' inter-arrival times, and the mathematical analysis provided in the paper lead to the intuitive form of $\overline{\mathcal{L}}_{\infty}$. In a later section, we will consider a particular application of interest and showcase the usefulness of the above-derived analytical results for the analysis of complex networks.
\begin{remark}
It is worth noting that although in our analysis we have focused on the case where one component makes up the system, our analysis can be extended to the case where the system is constituted of $N$ components. Particularly, as long as each component's operation is independent of the others, similar results to Theorem \ref{theoremfinal} can be obtained for the average expected utility not satisfied
\begin{equation}
\overline{\mathcal{L}}=\lim_{T\to+\infty} \:\frac{1}{NT}\sum_{i=1}^{N}\int_{0}^{T}U^{i}(t)W^{i}(t)dt,
\end{equation}
where $U^{i}(t)$ is the utility at time $t$ of component $i$ and $W^{i}(t)$ is a binary random variable indicating if component $i$ is suffering from an anomaly. As will be seen in Section \ref{subsectionconvergencemulti}, the consideration of large-scale systems (i.e., $N\gg 1$) will have a beneficial effect on the convergence speed of $\overline{\mathcal{L}}$.
\label{remarknetwork}
\end{remark}

\section{Generalizations of the Analysis}
\label{generalizationss}
The goal of this section is to provide extensions of our theoretical analysis to more general scenarios that could arise in practical situations. 
\subsection{Relaxing the Periodicity Assumption}
In practice, the utility function $U(t)$ does not exhibit perfect periodicity but, rather, fluctuates slightly around a regular period. These fluctuations are caused by unpredictable traffic demand and user activity. To account for this, we model the utility function $U(t)$ as a stochastic process that randomly fluctuates at each time instant $t$. Specifically, we write $U(t)$ as
\begin{equation}
U(t) = U'(t)+ B(t),\quad t\geq0,
\label{rewritingfunctionBt}
\end{equation}
where $U'(t)$ is a bounded, non-negative, periodic function with period $p$, and $B(t)$ is a stochastic process. The random fluctuations in $U(t)$ are determined by the variance of the process $B(t)$. In Section \ref{subsectionconvergencemulti}, we will use data from a real deployed network to show that these fluctuations are typically small. Next, we suppose that $B(t)$ verifies the following assumption.
\begin{assumption}[Zero-Mean Weakly Dependent Wide-Sense Stationary Random Processes]
    We say that a stochastic process $B(t)$ is a zero-mean weakly dependent wide-sense stationary random process if 
    \begin{itemize}
\item $\mathbb{E}[B(t)]=0$ for all $t\in\mathbb{R}$.
\item $\mathrm{Cov}[B(t),B(u)]=\rho(t-u)$ for all $t,u\in\mathbb{R}$.
\item $\frac{1}{T}\int_{0}^{T}|\rho(\tau)|d\tau\xrightarrow{T\rightarrow\infty}0$.
    \end{itemize}
    \label{assumptionaboutbt}
\end{assumption}
There are many processes that satisfy this condition, including the zero-mean Gaussian white noise and the Ornstein-Uhlenbeck process. For these particular processes, we extend the results of Theorem \ref{theoremfinal} below.
\begin{theorem}
\label{theoremfinalnewprocess}
If (i) the anomalies' inter-arrival times are exponentially distributed, (ii) the maintenance times verify Assumption \ref{assumptionony}, and (iii) $B(t)$ verifies Assumption \ref{assumptionaboutbt}, then the expected loss $\overline{\mathcal{L}}$ converges in probability to 
\begin{equation}
\overline{\mathcal{L}}_{\infty}=\frac{\mathbb{E}[Y]\overline{U}}{\mathbb{E}[X]+\mathbb{E}[Y]},
\end{equation}
where $\overline{U}=\frac{1}{p}\int_{0}^{p}U'(t)dt$.
\end{theorem}
\begin{proof}
The proof involves using Assumption \ref{assumptionaboutbt}, integral manipulations, and Chebyshev inequality to show that $B(t)$ exhibits ergodic properties. See Appendix \ref{appendixtheoremfinalnewprocess} for details. 
\end{proof}
\subsection{General Inter-Arrival Distributions}
In our system model reported in Section \ref{systmodel}, we have assumed that the anomalies' inter-arrival times are exponentially distributed. The reasons behind this assumption are twofold: 1) the data that we gathered from a large-scale deployed cellular network suggest that the exponential distribution is a good fit, as will be illustrated in Section \ref{subsectiondistributions}, and 2) the vast literature on reliability confirms that exponential distribution is found to be a good representative in many cases (e.g., \cite[Chapter~3]{doi:https://doi.org/10.1002/9781118841716.ch3},\cite{OHRING1995747}). Additionally, this assumption has allowed us to further understand the convergence rate reported in Theorem \ref{theoremuniform} and interpret the effect of the period $p$ and the anomalies' rate $\lambda$ on the convergence. Nevertheless, in this section, we aim to generalize our analytical results to a large family of distributions, particularly the family of bounded distributions. To that end, we present our results below.
\begin{theorem}
Let $f_{X}(x)$ denote the probability distribution function of the anomaly inter-arrival time $X_j$ for $j\in\mathbb{N}^*$. If there exists a finite $M>0$ such that
\begin{equation}
    f_{X}(x)\leq M,\quad x\in[0,\infty),
\end{equation}
then $f_{X^{[p]}_j}(x)$ converges uniformly to the uniform distribution on $[0,p]$. Precisely, there exists two constants $\alpha<1$ and $C>0$ such that
\begin{equation}
\sup_{x\in[0,p]} |f_{X^{[p]}_j}(x)-\frac{1}{p}|\leq \frac{C\alpha^{j}}{p},
\end{equation}
where
\begin{equation}
\alpha=\sup_{n\in\mathbb{N}^*}|\hat{g}(n)|,
\end{equation}
\begin{equation}
C=\frac{\sum_{n=1}^{\infty}|\hat{g}(n)|^{2}}{\alpha^2},
\end{equation}
and
\begin{equation}
\hat{g}(n)=\int_{0}^{\infty}f_{X}(x)e^{-2\pi i\frac{n}{p}x}dx\overset{\Delta}{=}F_X(\frac{n}{p}),
\end{equation}
where $F_X(\cdot)$ is the Fourier transform of $f_{X}(\cdot)$.
\label{theoremfouriergenerall}
\end{theorem}
\begin{proof}
The proof follows steps similar to those provided in Theorem \ref{theoremuniform}'s proof reported in Appendix \ref{prooffourrier}. For completeness, the details of the proof are reported in Appendix \ref{prooffourriergeneral}.
\end{proof}
Knowing that the rest of the mathematical analysis reported in Section \ref{mathanalysis} holds given the above convergence results, we can conclude that the results of Theorem \ref{theoremfinal} also hold for this general family of probability distributions. Similar conclusions can be made for %Theorem \ref{theoremfinalalmostperiodic} and
Theorem \ref{theoremfinalnewprocess}. In essence, if the anomalies' inter-arrival time distribution is bounded, and the remaining conditions of Theorem \ref{theoremfinal} and \ref{theoremfinalnewprocess} are verified, then the expected loss $\overline{\mathcal{L}}$ converges in probability to 
\begin{equation}
\overline{\mathcal{L}}_{\infty}=\frac{\mathbb{E}[Y]\overline{U}}{\mathbb{E}[X]+\mathbb{E}[Y]}.
\end{equation}

\section{Use case: A Large Cellular Network}
\label{usecasescellular}
In this section, our aim is to verify the key assumptions adopted in our analysis, corroborate our theoretical findings, and shed light on various challenges found in our implementations. To do so, we leverage the data that we have gathered from a large long term evolution (LTE) cellular network consisting of 660 cellular base stations and serving approximately 22k users. The data gathered span two months and consist of cell-level key performance indicators (KPIs), along with troubleshooting tickets generated by the network's operator and alarms information registered by the base stations. For confidentiality reasons, the data have been scaled when necessary. 
\subsection{Periodicity of Utility Functions}
\label{subsectionutility}
The first assumption we verify in this section, given its importance to our theoretical analysis, concerns the periodicity of the utility function of the cellular sites. Typically, the utility function of a base station is taken as the amount of traffic that it provides to the customers. In other words, the utility function at time $t$ is nothing but the user demand typically served by the base station at this time instant. In Fig. \ref{utilityoriginal}, we illustrate the average traffic served by a particular base station at each hour of the week. As seen in the figure, a noticeable periodic trend can be witnessed, although minor violations of the periodicity can be observed. However, given that these violations remain minor, and by recalling Theorem \ref{theoremfinalnewprocess}, we can conclude that the periodicity assumption comes at a minor penalty. 
\begin{figure}[!ht]
\centering
\includegraphics[width=.99\linewidth]{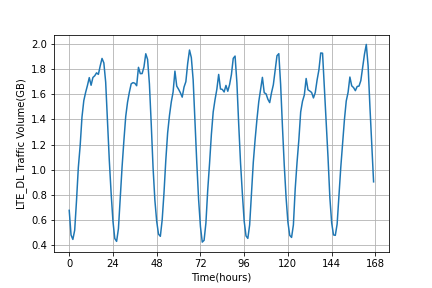}
\caption{LTE traffic demand.}
\label{utilityoriginal}
\end{figure}
\subsection{Anomalies' Key Distributions}
\label{subsectiondistributions}
Next, our goal is to investigate the distribution of both the anomalies' inter-arrival and maintenance times. To do so, we leverage the troubleshooting tickets data provided by the network's operator along with alarms information registered by the base stations. In essence, when an anomaly takes place at any base station in the network, the base station raises an alarm, and, consequently, a trouble ticket is issued by the operator. This ticket contains details about the anomaly (e.g., anomaly ID), its location, its occurrence time, and (eventually) its resolution time. As seen in Fig \ref{interarrivaldist}, the distribution of the anomalies' inter-arrival times in the overall network is very close to an exponential distribution of rate $\lambda=12.6$ anomalies/hour. Supposing that the base stations are all identical, and given the splitting property of Poisson processes \cite{bertsekastsitsiklis2008}, we can conclude that the anomalies' rate for each base station is $\lambda=0.019$ anomalies/hour.  
\begin{figure}[!ht]
\centering
\includegraphics[width=.99\linewidth]{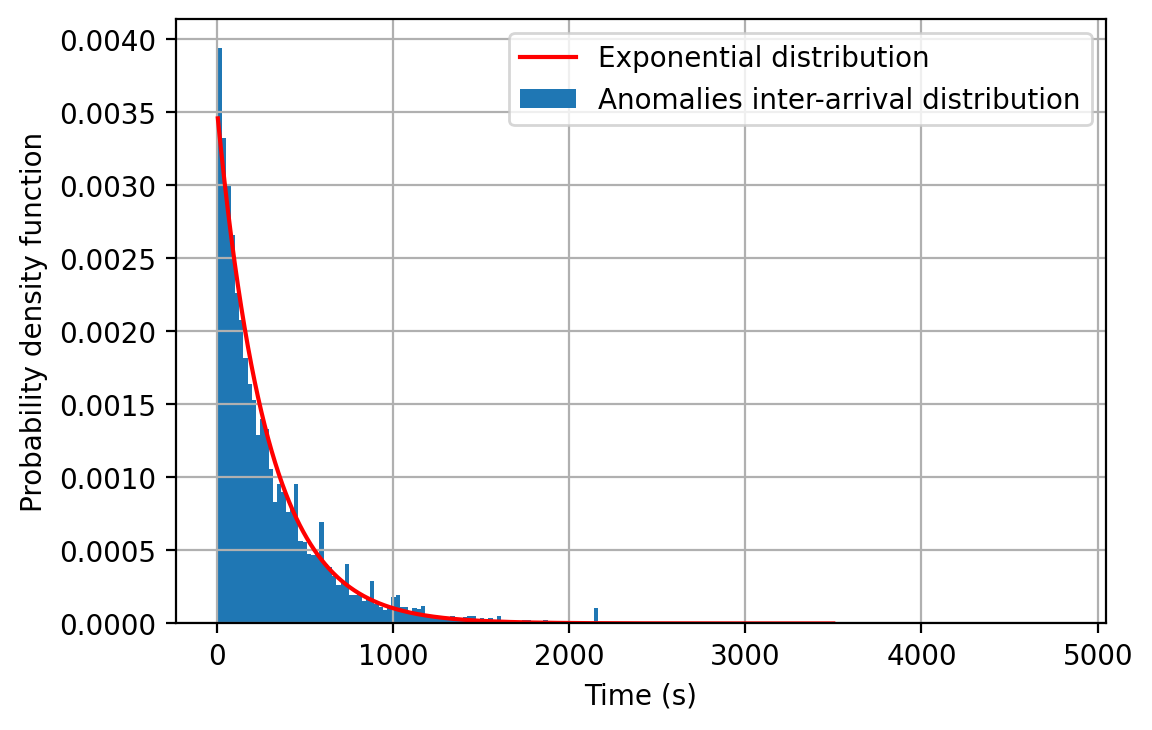}
\caption{Anomalies' inter-arrival times distribution.}
\label{interarrivaldist}
\end{figure}\\
On another note, we report in Fig. \ref{maintenancetimesss} an extract of the anomalies' maintenance time distribution. As one can see, this distribution is far from regular. Specifically, we can see that a part of the anomalies are resolved almost instantaneously by the network itself. On the other hand, other anomalies require either remote or on-site interventions that take longer time (hours, days, and sometimes weeks). Modeling such a distribution is a challenging task. However, we recall that the results reported in Section \ref{mathanalysis} hold for any general maintenance time distribution. In fact, all we need to characterize the expected utility not satisfied is the average maintenance duration, which puts into perspective the generality of our results and their practical usefulness. To that end, using the troubleshooting tickets data, we can conclude that the average maintenance time is 2 hours and 8 minutes.  
\begin{figure}[htb]
\centering
\includegraphics[width=.99\linewidth]{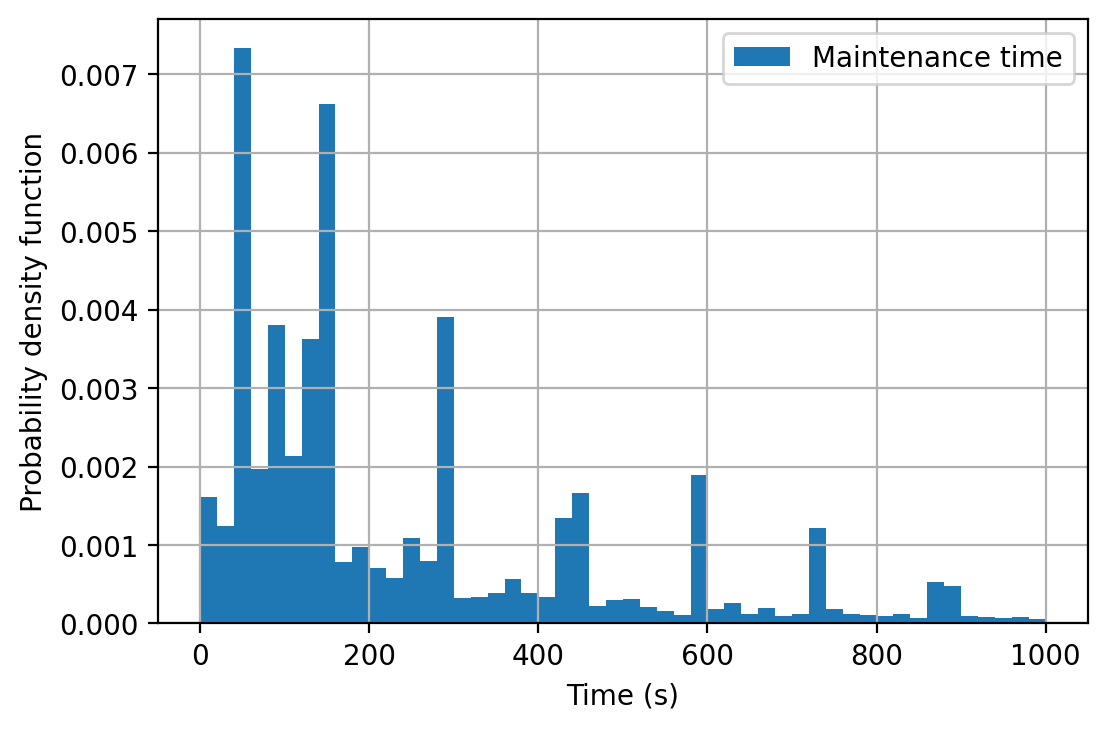}
\caption{Extract of the maintenance time distribution.}
\label{maintenancetimesss}
\end{figure}\\
With the above results in mind, along with the periodic nature of users' traffic as was illustrated in the previous section, we can conclude that the results of Theorem \ref{theoremfinal} can be used to find the expected utility not satisfied in the network. Particularly, knowing that the average traffic per hour for each base station is $\overline{U}=1.55\textnormal{ GBs/hour}$, the operator can, then, conclude that the expected traffic demand not satisfied in the entire network is
\begin{equation}
\overline{\mathcal{L}}_{\infty}=660\times\frac{2.13}{\frac{1}{0.019}+2.13}\times1.55\simeq39\:\: \textnormal{GBs/hour}
\end{equation}
All in all, we can conclude that the network loses on average around $\frac{2.13}{\frac{1}{0.019}+2.13}\simeq3.8\%$ of its traffic due to the various anomalies that may occur. This reliability score can be used by the operator to assess the network's performance, with respect to its business objective. Further planning and network upgrades can, then, take place if necessary, and a re-evaluation of the score can be done to conclude on the efficacy of the proposed upgrades.
\subsection{Convergence Speed: Single Cell}
Thus far, we have been interested in verifying the key assumptions adopted in our theoretical analysis and characterizing the distributions needed to evaluate the theoretical limit of the expected demand not satisfied $\overline{\mathcal{L}}$, denoted by $\overline{\mathcal{L}_{\infty}}$. However, investigating how fast we converge to this limit $\overline{\mathcal{L}_{\infty}}$ is of paramount interest. To that end, let 
\begin{equation}
\overline{\mathcal{L}}_T=\frac{1}{T}\int_{0}^{T}U(t)W(t)dt
\label{expectedlossfinitetime}
\end{equation}
denote the average traffic demand not satisfied in the time interval $[0,T]$. Equivalently, as we have done in eq. (\ref{expressiontaba3lasess}), we can define the average demand not satisfied in the anomaly-repair cycles $[0,n]$ as
\begin{equation}
\overline{\mathcal{L}}_n=\frac{\frac{1}{n}\sum_{j=1}^{n}\int_{D_j ^{[p]}-Y_j}^{D^{[p]}_j}U(t)dt}{\frac{1}{n}\sum_{j=1}^{n}(X_j+Y_j)}.
\end{equation}
With these definitions in mind, our aim is to find at what stage $n^*$ the relative error $
\frac{|\overline{\mathcal{L}}_n-\overline{\mathcal{L}_{\infty}}|}{\overline{\mathcal{L}}_n}$ becomes smaller than $10\%$ for $n\geq n^*$. To do so, we first consider that the inter-arrival time between anomalies follows an exponential distribution with rate $\lambda = 0.019$ hours$^{-1}$. As for the maintenance time, to facilitate our convergence speed study, we consider that the maintenance time is exponentially distributed with rate $\mu=0.47$ hours$^{-1}$, which is in accordance with the average maintenance time duration observed in the data. Note that the insights we are about to provide would still hold for other general maintenance time distributions. Finally, concerning the utility function, we suppose that the utility function is perfectly periodic and can be written as
\begin{equation}
U(t)=1.75\sin(\frac{2\pi t}{p})+3, \quad t\geq0,
\label{unumerical}
\end{equation}
where $p=24$ hours. The above function mimics very well the utility function illustrated in Fig. \ref{utilityoriginal}. We report our results in Fig. \ref{convergenceonecell}. As can be seen, the loss function $\overline{\mathcal{L}}_n$ converges to $\overline{\mathcal{L}}_{\infty}$, as Theorem \ref{theoremfinal} predicts. Particularly, it takes on average $n^*=350$ anomaly-repair cycles for the relative error to become smaller than $10\%$. Given the average time of each cycle, we can conclude that the convergence time can be quite slow. The reason behind this is that although the convergence of $D^{[p]}_j$ to a uniform distribution happens exponentially fast, as shown in Proposition \ref{theoremadduniform}, the overall convergence of $\overline{\mathcal{L}}_n$ to $\overline{\mathcal{L}_{\infty}}$ reported in Theorem \ref{theoremfinal} is not exponentially fast. Particularly, as can be seen in Appendix \ref{appendixtheoremfinal}, the convergence speed is inversely proportional to $n$ and depends heavily on the correlation between the different quantities reported in Proposition \ref{covarianceproposition} and the variance of the random variables involved. However, as we will show in the next section, this convergence speed becomes substantially faster when a large-scale network is considered, rather than an individual base station.
\begin{figure}[!ht]
\centering
\includegraphics[width=.99\linewidth]{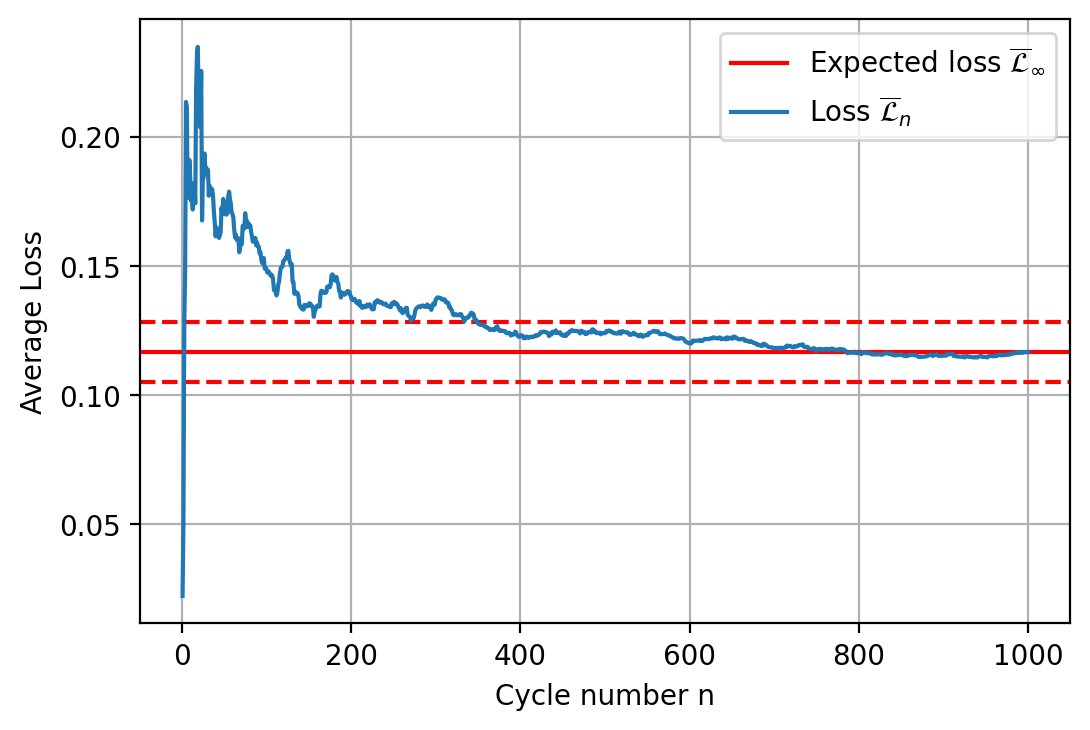}
\caption{Illustration of the convergence of the cellular network with one cell}
\label{convergenceonecell}
\end{figure}
\subsection{Convergence Speed: Multiple Cells}
\label{subsectionconvergencemulti}
In practice, the network's operator is interested in evaluating the expected utility loss of the whole network rather than that of a particular cell. Specifically, if the network is made of $N$ cells, then we are interested in the speed of convergence of 
\begin{equation}
\overline{\mathcal{L}}_T=\:\frac{1}{NT}\sum_{i=1}^{N}\int_{0}^{T}U^{i}(t)W^{i}(t)dt,
\end{equation}
to its limit $\overline{\mathcal{L}}_{\infty}$, where $U^{i}(t)$ is the utility at time $t$ of cell $i$ and $W^{i}(t)$ is a binary random variable indicating if cell $i$ is suffering from an anomaly. To investigate this convergence speed, we adopt the same settings of the previous subsection and compare the multi-cell case to the single-cell one. Note that since anomalies may occur at different instants at the various cells, we report the average convergence speed in time units rather than anomaly-repair cycles. The results are reported in Fig. \ref{cvcompare}. 
\begin{figure}[!ht]
\centering
\includegraphics[width=.99\linewidth]{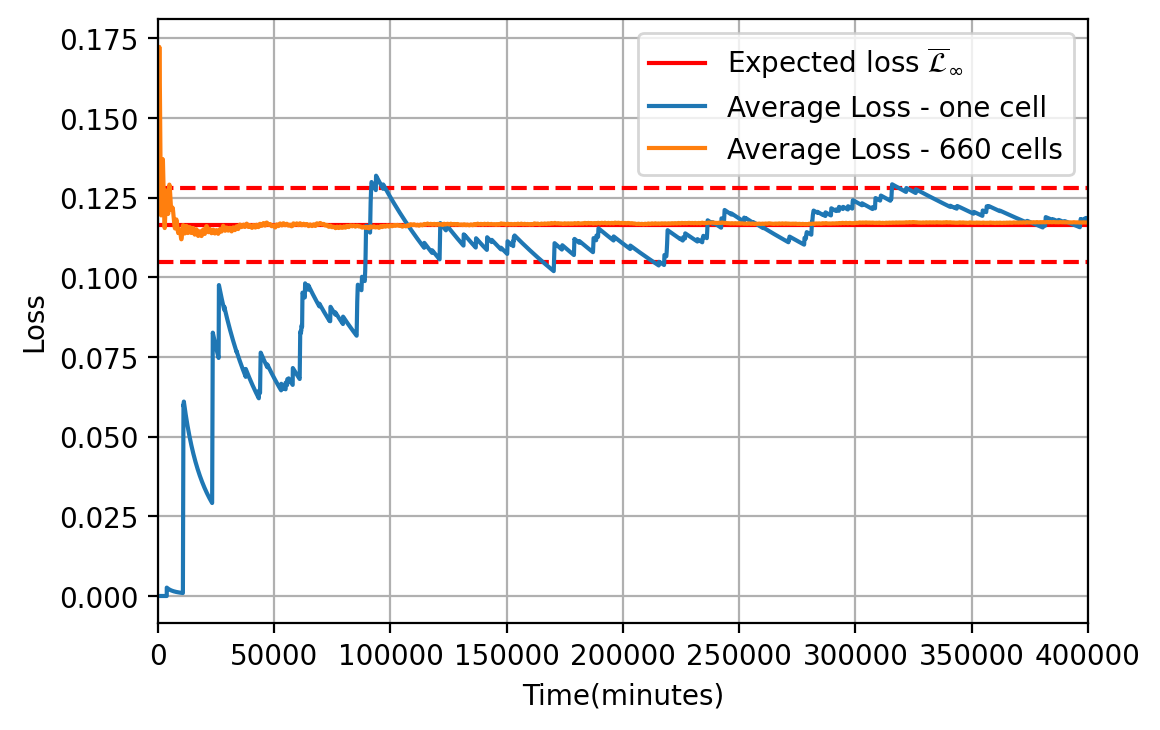}
\caption{Comparison of the convergence time in single and multi-cell scenarios.}
\label{cvcompare}
\end{figure}\\
As can be seen, the average convergence speed for the whole network's average utility loss is around $41$ hours, which is significantly faster than the single-cell one. In other words, by letting the network run for a couple of days, our evaluation of the expected utility not satisfied through the simple theoretical formula we provided in Theorem \ref{theoremfinal} would already be very accurate. Given that networks are typically constituted of a large number of sub-systems (e.g., cells, routers, etc.), we can conclude that this convergence speed is particularly appealing. To understand this trend, we note that the large number of cells leads to a spatial averaging effect. Particularly, as more cells are added, the variance reported in Appendix \ref{appendixtheoremfinal} will be reduced, thus greatly accelerating the convergence speed. The question that remains is the following: how big does the network need to be to obtain such fast convergence? To answer that question, we simulate $100$ different realizations of the network's operation for a variety of network sizes. The results are reported in Fig. \ref{cvwithncell}. As can be observed, the average convergence time decreases with the number of cells, reaching as low as $1$ week for a network size of $100$ cells. 
%Accordingly, even for a system constituted of $100$ components, the convergence's speed is already high.
\begin{figure}[!ht]
\centering
\includegraphics[width=.99\linewidth]{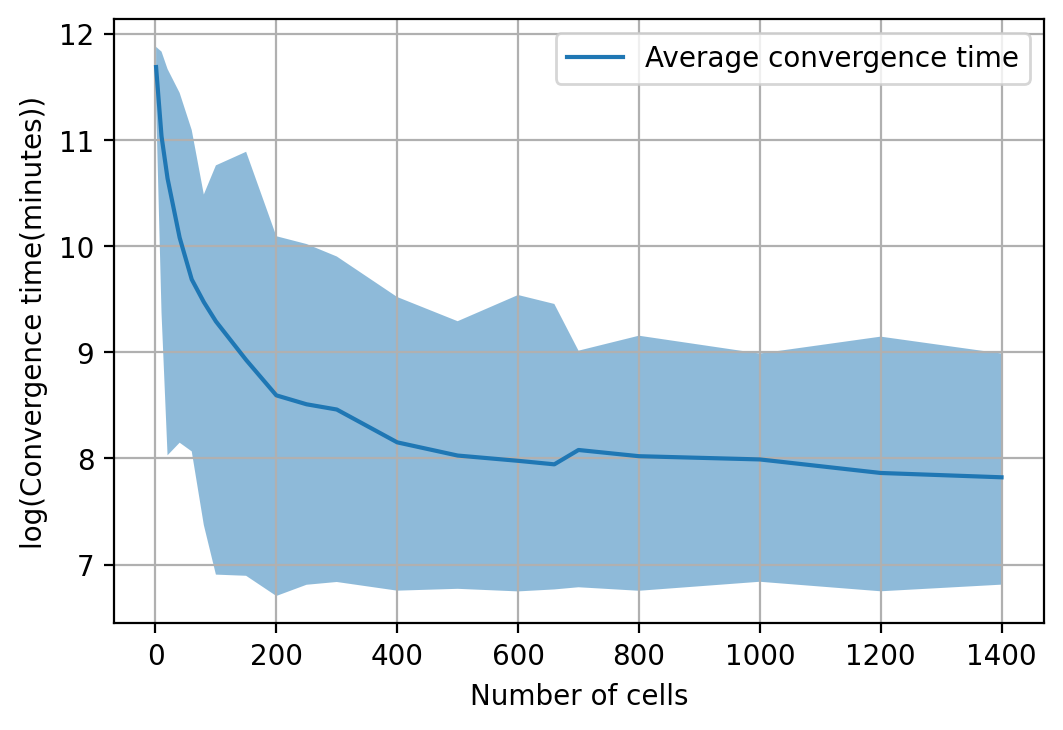}
\caption{Evolution of the convergence time with the number of cells. }
\label{cvwithncell}
\end{figure}\\
Finally, we examine the impact of violating the assumption of perfect periodicity on the convergence speed of $\overline{\mathcal{L}}_T$. To do this, we compare the perfect periodicity case to the scenario where $U(t)$ in eq. (\ref{unumerical}) is corrupted by an Ornstein-Uhlenbeck process $B(t)$. Specifically, $B(t)$ satisfies the following stochastic differential equation:
\begin{equation}
dB(t)=-\theta B(t)dt+\sigma dR(t),
\end{equation}
where $R(t)$ is the Wiener process. By setting $\theta=1$ and $\sigma=0.01$, we plot the relative error $\frac{|\overline{\mathcal{L}}_T-\overline{\mathcal{L}}_{\infty}|}{\overline{\mathcal{L}}_T}$ as a function of $T$. As shown in Fig. \ref{cvnoise}, the relative error between the expected loss and its limit calculated by our formula quickly converges to zero. In fact, the convergence speed of the two cases is similar, with both reaching the $10\%$ mark at around 41 hours. This confirms the results of Theorem \ref{theoremfinalnewprocess}, which states that even with minor, random fluctuations from the perfect periodicity regime, our convergence results still hold.
\begin{figure}[!ht]
\centering
\includegraphics[width=.99\linewidth]{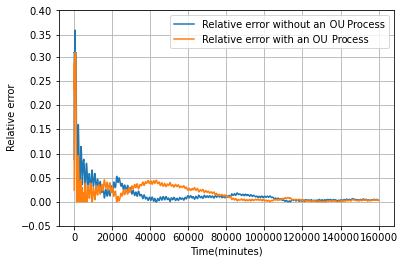}
\caption{Evolution of the relative error between the loss and its expected limit as a function of $T$.}
\label{cvnoise}
\end{figure}
\subsection{Non-Stationarity}
Thus far, we have assumed in our analysis that the inter-arrival times and maintenance times of the anomalies follow stationary distributions. This means that the statistics of these distributions, such as the expected values $\mathbb{E}[X]$ and $\mathbb{E}[Y]$, do not change over time. However, it is important to verify the accuracy of this assumption. To do so, we can use our real data to plot 
\begin{equation}
\Delta=\frac{\mathbb{E}[Y]}{\mathbb{E}[X]+\mathbb{E}[Y]}=1-\textnormal{Availability}
\end{equation}
to see if it remains constant over time. Our attention is on $\Delta$, as the theoretical limit $\mathcal{L}_{\infty}$ depends on the aforementioned distributions through $\Delta$. Fig. \ref{percentageabnormalcell} demonstrates that $\Delta$ has minimal fluctuations, causing minor violations of the stationarity assumption. Nevertheless, these fluctuations have minimal impact on the convergence time, as can be seen in Fig. \ref{estimationerror}. The relative error $\frac{|\overline{\mathcal{L}}_T-\overline{\mathcal{L}}_{\infty}|}{\overline{\mathcal{L}}_T}$ fluctuates minimally and remains close to zero. 
\begin{figure}[!ht]
\centering
\includegraphics[width=.99\linewidth]{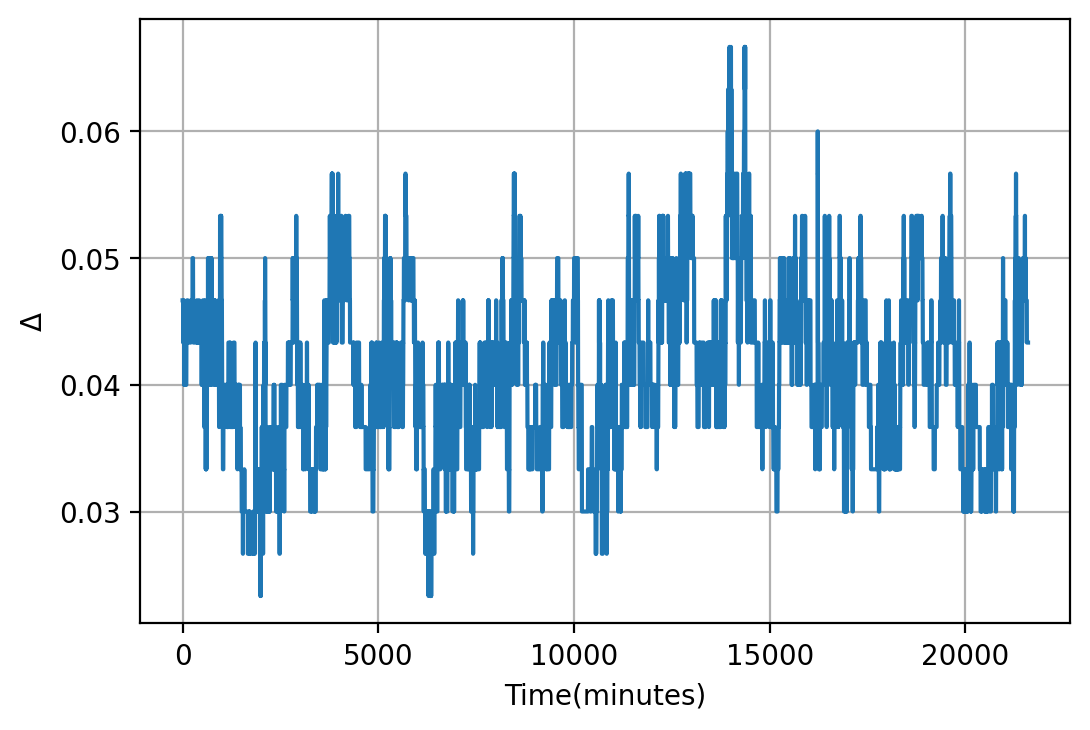}
\caption{Variation of $\Delta$ as a function of time.}
\label{percentageabnormalcell}
\end{figure}
\begin{figure}[!ht]
\centering
\includegraphics[width=.99\linewidth]{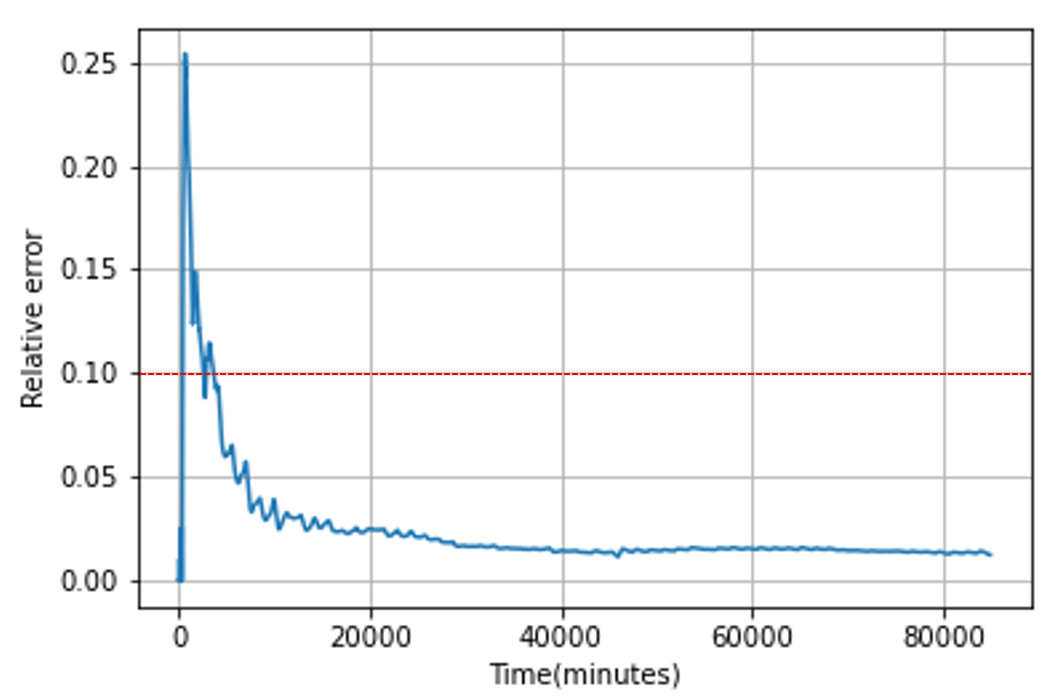}
\caption{Variation of the relative error as a function of time.}
\label{estimationerror}
\end{figure}

\section{Conclusions and Future Work}
\label{conclusionsss}
In this paper, we have considered the evaluation of the reliability of large networked systems with respect to a periodic time-dependent utility function related to the system's service performance over time. Under the assumption of exponential anomalies' inter-arrival times and general distributions of maintenance time duration, we have leveraged the periodicity of the utility function to derive the expected utility loss due to the system's anomalies. In these settings, we have shown that the expected utility loss converges in probability to a simple form. We have also extended our results to the case where the periodicity of the utility is slightly violated and the distribution of the inter-arrival times of anomalies is general. We applied our analysis to a cellular network use case with real field data gathered from a deployed cellular network and demonstrated its usefulness. Additionally, we verified our key assumptions using the available data and investigated the interplay between various network parameters.
\color{black}
\bibliographystyle{IEEEtran}
\bibliography{References}
\appendices
\section{Proof of Theorem \ref{theoremuniform}}
\label{prooffourrier}
Our proof revolves around a Fourier analysis of the distribution of $X^{[p]}_j$. To proceed in this direction, let us first note that by definition
\begin{equation}
\sum_{k=1}^{j}X_k=mp+X^{[p]}_j=p(m+\frac{X^{[p]}_j}{p}),
\end{equation}
where $m\in\mathbb{N}$. By dividing both sides by $\frac{1}{p}$, we can conclude that
\begin{equation}
X^{[p]}_j=pR_j^{[1]},
\end{equation}
where $R_j^{[1]}$ is the remainder of the Euclidean division of $R_j=\sum_{k=1}^{j}\frac{X_k}{p}$ by $1$. To that end, to characterize the distribution of $X^{[p]}_j$, we can start by studying that of $R_j^{[1]}$. With this in mind, the first step of our analysis in this direction consists of showing the existence of an order $j^*$ such that
\begin{equation}
\sup_{r\in[0,1]}f_{R_{j^*}^{[1]}}(r)<\infty,
\label{conditionbundedfirst}
\end{equation}
where $f_{R_{j^*}^{[1]}}(r)$ is the probability density function of $R_{j^*}^{[1]}$.
The above condition is essential to our Fourier analysis as will be seen shortly. To find $j^*$, we first note that thanks to the scaling property of the exponential distribution, $R_1$ is exponentially distributed with rate $\lambda p$. Therefore, we have
\begin{equation}
\sup_{r\in[0,\infty)}f_{R_1}(r)=f_{R_1}(0)=\lambda p.
\end{equation}
Given that $f_{R_1}(r)$ is bounded, we can assert that the same holds for $f_{R_{1}^{[1]}}(r)$. Consequently, we can conclude that $j^*=1$. Next, let
\begin{equation}
\hat{g}(n)=\int_{0}^{\infty}f_{R_1}(r)e^{-2\pi inr}dr
\end{equation}
denote the nth Fourier coefficient of the probability density function of $R_1$. Knowing that $R_1$ is exponentially distributed, we obtain
\begin{align}
\hat{g}(n)&=\int_{0}^{\infty}\lambda p e^{-\lambda p r}e^{-2\pi inr}dr=\frac{\lambda p}{\lambda p+2\pi in}.
\end{align}
Given the above expression, we can conclude that $|\hat{g}(n)|$ is decreasing with $n$. Therefore, we have
\begin{equation}
\alpha=\sup_{n\in\mathbb{N}^*}|\hat{g}(n)|=|\hat{g}(1)|=\frac{\lambda p}{\sqrt{\lambda^2p^2+4\pi^2}}.
\end{equation}
Given that $\alpha<1$ and that condition (\ref{conditionbundedfirst}) holds, we can leverage \cite[Theorem~1]{schatte1984asymptotic} to deduce that 
\begin{equation}
\sup_{r\in[0,1]} |f_{R_j^{[1]}}(r)-1|\leq C\alpha^{j},
\label{resultslmhemfirst}
\end{equation}
where 
\begin{equation}
C=\alpha^{-2}\sum_{n=1}^{\infty}|\hat{g}(n)|^{2j^*}=\sum_{n=1}^{\infty}\frac{\lambda^2p^2+4\pi^2}{\lambda^2p^2+4\pi^2n^2}<\infty.
\end{equation}
Note that the convergence of $C$ is due to the fact that the term $n^2$ exists in the denominator. Lastly, given that $X_j^{[p]}=pR_j^{[1]}$, we can conclude that
\begin{equation}
f_{X^{[p]}_j}(x)=\frac{1}{p}f_{R_j^{[1]}}(\frac{x}{p}).
    \label{scaling}
\end{equation}
With this in mind, and using the results of eq. (\ref{resultslmhemfirst}), we can conclude our proof.
\section{Proof of Lemma \ref{addingsmoothesout}}
\label{proofsum}
Due to the independence between $A$ and $B$, we can derive the probability density function $f_{Z}(z)$ for $z\in[0,p]$ as follows 
\begin{align}
f_{Z}(z)&=(f_{A}\circledast_{p} f_{B})(z)=\int_{0}^{z} f_{A}(x)f_{B}(z-x)dx\nonumber\\&+\int_{z}^{p} f_{A}(x)f_{B}(z+p-x)dx,
\end{align}
where $\circledast_{p}$ indicates that the convolution is done modulo $p$. Next, we investigate the gap between the density function $f_{Z}(\cdot)$ and the uniform distribution on $[0,p]$. Precisely,
\begin{align}
|f_{Z}(z)-\frac{1}{p}|&\overset{(a)}=|\int_{0}^{z} (f_{A}(x)-\frac{1}{p})f_{B}(z-x)dx\nonumber\\&+\int_{z}^{p} (f_{A}(x)-\frac{1}{p})f_{B}(z+p-x)dx|\nonumber\\&\overset{(b)}{\leq}\int_{0}^{z}|f_{A}(x)-\frac{1}{p}|f_{B}(z-x)dx
\nonumber\\&+\int_{z}^{p} |f_{A}(x)-\frac{1}{p}|f_{B}(z+p-x)dx\nonumber\\&\overset{(c)}\leq\sup_{x\in[0,p]}|f_{A}(x)-\frac{1}{p}|(\int_{0}^{z}f_{B}(z-x)dx\nonumber\\&+\int_{z}^{p}f_{B}(z+p-x)dx)                                                                                                                                                                                                                                                                                                                                                                                                                                                                                                                                                                                                                                                                                                                                                                                                                                                                                                                                                                                                                                                                                                                                                                                                                                                                                                                                                                                                                                                                                                                                                                                                                                                                                                                                                                                                                                                                                                                                                                                                                                                                                                                                                                                                                                                                                                                                                                                                                                                                                                                                                                                                              \nonumber\\&\overset{(d)}=\sup_{x\in[0,p]}|f_{A}(x)-\frac{1}{p}|,
\end{align}
where $(a)$ is the result of writing $\frac{1}{p}$ as an integral over the density function $f_{B}(\cdot)$, $(b)$ is deduced from the triangular inequality, $(c)$ can be concluded from the definition of the supremum, and $(d)$ comes from the fact that $f_{B}(\cdot)$ is a density function. By taking the supremum on the left hand side, we can conclude the lemma. 

\section{Proof of Proposition \ref{theoremfinal}}
\label{proofcesaro}
To obtain our results, we leverage the Stolz–Cesàro theorem. To that end, let us define the sequences $a_n=\sum_{j=1}^{n}\mathbb{E}[I_j]$ and $b_n=n$. We have
\begin{equation}
\lim_{n\to+\infty}\frac{a_{n+1}-a_{n}}{b_{n+1}-b_{n}}=\lim_{n\to+\infty}\frac{\mathbb{E}[I_{n+1}]}{1}.
\end{equation}
Let us now define $\overline{I}$ as follows
\begin{align}
\overline{I}=\mathbb{E}_{Y,V}[\int_{V-Y}^{V}U(t)dt],
\end{align}
%\nonumber\\&=\frac{\mathbb{E}[Y]}{\frac{1}{\lambda}+\mathbb{E}[Y]}\frac{1}{p}\int_{0}^{p}U(t)dt, 
where $V$ is a RV uniformly distributed on $[0,p]$. In essence, $\overline{I}$ is the expectation of $I_j$ if $D_j^{[p]}$ is uniformly distributed. Next, we develop the expression above as follows
\begin{equation}
\overline{I}=\int_{0}^{\infty}\int_{0}^{p}\int_{v-y}^{v}U(t)\frac{1}{p}f_{Y}(y)dtdvdy.
\end{equation}
where $f_{Y}(y)$ is the probability density function of the RVs $Y_j$ for $j\in\mathbb{N}^*$. To further simplify the expression above, we proceed with a change of variable. Precisely, we let $s=t-v$. Consequently, we end up with
\begin{align}
\overline{I}&=\int_{0}^{\infty}\int_{0}^{p}\int_{-y}^{0}U(s+v)\frac{1}{p}f_{Y}(y)dsdvdy\nonumber\\&\overset{(a)}{=}\int_{0}^{\infty}\int_{-y}^{0}\int_{0}^{p}U(s+v)\frac{1}{p}f_{Y}(y)dvdsdy\nonumber\\&\overset{(b)}{=}\int_{0}^{\infty}\int_{-y}^{0}\int_{0}^{p}U(v)\frac{1}{p}f_{Y}(y)dvdsdy\nonumber\\&=\mathbb{E}[Y]\frac{1}{p}\int_{0}^{p}U(v)dv=\mathbb{E}[Y]\overline{U},
\label{therewritingofibar}
\end{align}
where $(a)$ is the result of changing the order of integration, and $(b)$ is due to the periodicity of $U(t)$. Now, let us write $I_{n+1}$ in a more convenient form as follows
\begin{equation}                                                                                                                                                                                                                                                                                                                                                                                                                                                                                                                                                                                                                                                                                                                                        
I_{n+1}=\int_{D_{n+1}^{[p]}-Y_j}^{D^{[p]}_{n+1}}U(t)dt=\int_{[D_{n}+X_{n+1}]\text{ mod } p}^{[D_{n}+X_{n+1}]\text{ mod } p + Y_{n+1}}U(t)dt.
\label{eqwithouD}
\end{equation}
Then, we can proceed as done in eq. (\ref{therewritingofibar}) to obtain
\begin{equation}                                                                                                                                                                                                                                                                                                                                                                                                                                                                                                                                                                                                                                                                                                                                        
\mathbb{E}[I_{n+1}]=\int_{0}^{\infty}\int_{-y}^{0}\int_{0}^{p}U(s+y+v)f_{V}(v)f_{Y}(y)dvdsdy,
\end{equation}
where $f_{V}(v)$ is the probability density function of $[D_{n}+X_{n+1}]\text{ mod } p$. By leveraging Lemma \ref{addingsmoothesout} and Proposition \ref{theoremadduniform}, we can conclude that
\begin{equation}
\sup_{v\in[0,p]} |f_{V}(v)-\frac{1}{p}|\leq \frac{C\alpha^{n}}{p}.
\end{equation}
Consequently, we have $f_{V}(v)\leq \dfrac{C\alpha^{n}}{p} +\dfrac{1}{p}$. Knowing this, we can now investigate the difference between $\overline{I}$ and $\lim_{n\to+\infty}\mathbb{E}[I_{n+1}]$. Given all what was presented above, we obtain
\begin{equation}
|\mathbb{E}[I_{n+1}]-\overline{I}|\leq\frac{C\alpha^{n}}{p}\mathbb{E}[Y]\overline{U}\xrightarrow{n \to +\infty}0,
\label{convergenceibar}
\end{equation}
given that $\mathbb{E}[Y],\overline{U}$ are finite and $\alpha<1$. Finally, given that $b_n$ is a strictly monotone and divergent sequence, we can conclude from the Stolz–Cesàro theorem that $S_n$ converges to $\overline{I}$ for large $n$.

\section{Proof of Proposition \ref{jointresults}}
\label{appendixjoint}
To start our proof, we first note that by the definition of $D^{[p]}_{j}$ reported in eq. (\ref{rewritedj}) and the independence of $X_i$ and $Y_i$, we have
\begin{equation}
f_{D^{[p]}_{j},D^{[p]}_{j+k}}(z_j,z_{j+k})=f_{D^{[p]}_{j}}(z_j)f_{D^{[p]}_{k}}([z_{j+k}-z_j]\text{ mod }p).
\end{equation}
Next, thanks to Lemma \ref{addingsmoothesout}, and by leveraging the triangular inequality, we can conclude that
\begin{align}
\sup_{z\in[0,p]} |f_{D^{[p]}_j}(z)-\frac{1}{p}|&\leq \sup_{x\in[0,p]} |f_{X^{[p]}_1}(x)-\frac{1}{p}|\nonumber\\&\leq \sup_{x\in[0,p]}f_{X^{[p]}_1}(x)+\frac{1}{p}.
\end{align}
We recall that $X_1$ is an exponential random variable of rate $\lambda$. Hence, the supremum of its probability density function is equal to $\lambda$. Given that the probability density function of $X_1$ is bounded, we can conclude that there exists $M$ such that $f_{X^{[p]}_1}(x)\leq M$. With this in mind, and by leveraging the triangular inequality, we can deduce that
\begin{align}
&|f_{D^{[p]}_{j}}(\cdot)f_{D^{[p]}_{k}}(\cdot)-\frac{1}{p^2}|=| (f_{D^{[p]}_{j}}(\cdot)-\frac{1}{p})f_{D^{[p]}_{k}}(\cdot)+\nonumber\\&\frac{1}{p}(f_{D^{[p]}_{k}}(\cdot)-\frac{1}{p})|\leq  (M+\frac{1}{p})\frac{C\alpha^{j}}{p}+\frac{1}{p}\frac{C\alpha^{k}}{p}.
\end{align}
Finally, by letting $C'=M+\frac{1}{p}$, we can conclude the proposition.

\section{Proof of Proposition \ref{covarianceproposition}}
\label{appendixcovarianceproposition}
By definition, we have
\begin{equation}
\mathrm{Cov}[I_j,I_{j+k}]=\mathbb{E}[I_jI_{j+k}]-\mathbb{E}[I_j]\mathbb{E}[I_{j+k}].
\end{equation}
Let us first investigate the term $\mathbb{E}[I_jI_{j+k}]$. As was done in eq. (\ref{eqwithouD}), we can rewrite $\mathbb{E}[I_jI_{j+k}]$ as follows
\begin{align}
\mathbb{E}[I_jI_{j+k}]=&\mathbb{E}_{Y_j,Y_{j+k}}[\int_{0}^{p}\int_{0}^{p}\int_{z_j}^{z_j+Y_j}U(t)dt\times\int_{z_{j+k}}^{z_{j+k}+Y_{j+k}}\nonumber\\&U(t')dt' f_{Z_j^{[p]},Z_{j+k}^{[p]}}(z_j,z_{j+k})dz_jdz_{j+k}],
\end{align}
where $Z_j^{[p]}$ and $Z_{j+k}^{[p]}$ denote $[D_{j}+X_{j+1}]\text{ mod } p$ and $[D_{j+k}+X_{j+k+1}]\text{ mod } p$ respectively. By leveraging Lemma \ref{addingsmoothesout} and Proposition \ref{jointresults}, we obtain
\begin{align}
\mathbb{E}[I_jI_{j+k}]\leq&\mathbb{E}_{Y_j,Y_{j+k}}[\int_{0}^{p}\int_{0}^{p}\int_{z_j}^{z_j+Y_j}U(t)dt\times\int_{z_{j+k} }^{z_{j+k}+Y_{j+k}}\nonumber\\&U(t')dt'(\frac{1}{p^2}+C'C(\frac{\alpha^{j}}{p}+\frac{\alpha^{k}}{p})dz_jdz_{j+k}].
\end{align}
Given that the RVs $Y_j$ and $Y_{j+k}$ are i.i.d., we obtain
\begin{equation}
\mathbb{E}[I_jI_{j+k}]\leq \overline{I}^2(1+pC'C(\alpha^{j}+\alpha^{k})).
\end{equation}
Next, we need to investigate the second term $\mathbb{E}[I_j]\mathbb{E}[I_{j+k}]$. To do so, we recall from eq. (\ref{convergenceibar}) that
\begin{equation}
|\mathbb{E}[I_{j}]-\overline{I}|\leq\frac{C\alpha^{j-1}}{p}\overline{I}.
\end{equation}
With this in mind, we can obtain the results of the proposition, which concludes our proof.
\section{Proof of Theorem \ref{theoremfinal}}
\label{appendixtheoremfinal}
From the expression in eq. (\ref{expressiontaba3lasess}), we can rewrite $ \overline{\mathcal{L}}$ as
\begin{equation}
\overline{\mathcal{L}}=\lim_{n\to+\infty} \:\frac{\frac{1}{n}\sum_{j=1}^{n}I_j}{\frac{1}{n}\sum_{j=1}^{n}(X_j+Y_j)}.
\end{equation}
Let us start by investigating the numerator $\mathcal{N}_n=\frac{1}{n}\sum_{j=1}^{n}I_j$. To that end, let us consider the event
\begin{align}
\Pr(|\mathcal{N}_n-\overline{I}|\geq\epsilon)&\overset{(a)}\leq \Pr(|\mathcal{N}_n-S_n|\geq \frac{\epsilon}{2}\cup|S_n-\overline{I}|\geq\frac{\epsilon}{2})\nonumber\\ &\overset{(b)}\leq\Pr(|\mathcal{N}_n-S_n|\geq \frac{\epsilon}{2})+\Pr(|S_n-\overline{I}|\geq\frac{\epsilon}{2}),
\end{align}
where $(a)$ and $(b)$ are the results of the triangular inequality and the union bound, respectively. Note that the term $|S_n-\overline{I}|$ is deterministic, and from Proposition \ref{convergenceofthesum}, we know that this term tends to zero when $n$ is large. Therefore, what remains is to examine the first term $\Pr(|\mathcal{N}_n-S_n|\geq \frac{\epsilon}{2})$. To do so, we leverage Chebyshev's inequality to obtain the following upperbound
\begin{equation}
\Pr(|\mathcal{N}_n-S_n|\geq \frac{\epsilon}{2})\leq \frac{4\mathrm{Var}[\mathcal{N}_n]}{\epsilon^2}.
\end{equation}
However, we know that 
\begin{equation}
\mathrm{Var}[\mathcal{N}_n]=\frac{\displaystyle\sum_{j=1}^{n}\mathrm{Var}[I_j]+2\sum_{j=1}^{n-1}\sum_{j<k\leq n}\mathrm{Cov}[I_j,I_{k}]}{n^2}.
\end{equation}
To upperbound $\mathrm{Var}[\mathcal{N}_n]$, we first need to find an upperbound of $\mathrm{Var}[I_j]$. To do so, we consider in the following the term $\mathbb{E}[I_j^2]$. Precisely, we have
\begin{equation}
\mathbb{E}[I_j^2]=\mathbb{E}_{Y_j,D_j ^{[p]}}[(\int_{D_j ^{[p]}-Y_j}^{D^{[p]}_j}U(t)dt)^2]
\end{equation} 
Given that $U(t)\leq K$ for $t\geq0$, we obtain
\begin{equation}
\mathbb{E}[I_j^2]\leq K^2\mathbb{E}_{Y}[Y^2].
\end{equation}
With this in mind, and by using the results of eq. (\ref{convergenceibar}), we can conclude that
\begin{equation}
\mathrm{Var}[I_j]\leq K^2\mathbb{E}_{Y}[Y^2]-\overline{I}^2(1-\frac{C\alpha^{j-1}}{p})^2
\end{equation}
Now, using Proposition \ref{covarianceproposition} and with some algebraic manipulations, we can upperbound the variance as follows
\begin{align}
&\mathrm{Var}[\mathcal{N}_n]\leq \frac{K^2\mathbb{E}_{Y}[Y^2]-\overline{I}^2(1-\frac{C}{p})^2}{n}+\frac{2\overline{I}^2C}{n^2}(\nonumber\\&\frac{nC'p}{1-\alpha}+\frac{\alpha^2C'p}{(1-\alpha)^2}+\frac{n}{p(1-\alpha)}+\frac{\alpha^2}{p(1-\alpha)(1-\alpha^2)}).
\end{align}
Consequently, we can conclude that $\lim_{n\to+\infty} \Pr(|\mathcal{N}_n-S_n|\geq\frac{\epsilon}{2})=0$, and hence $\lim_{n\to+\infty} \Pr(|\mathcal{N}_n-\overline{I}|\geq\epsilon)=0$. As for the denominator $\mathcal{D}_n$, we know from the weak law of large numbers that $\frac{1}{n}\sum_{j=1}^{n}(X_j+Y_j)$ converges in probability to $\frac{1}{\lambda}+\mathbb{E}[Y]$. With this in mind, and given that $\Pr(\mathcal{D}_n=0)=0$, we can leverage the Mann–Wald theorem to conclude our theorem.
\section{Proof of Theorem \ref{theoremfouriergenerall}}
\label{prooffourriergeneral}
Our proof follows similar steps to those provided in Theorem \ref{theoremuniform}'s proof reported in Appendix \ref{prooffourrier}. Particularly, we first note that $X^{[p]}_j$ can be written as 
\begin{equation}
X^{[p]}_j=pR_j^{[1]},
\end{equation}
where $R_j^{[1]}$ is the remainder of Euclidean division of $R_j=\sum_{k=1}^{j}\frac{X_k}{p}$ by $1$. Now, similarly, we seek to prove the existence of an order $j^*$ such that
\begin{equation}
\sup_{r\in[0,1]}f_{R_{j^*}^{[1]}}(r)<\infty,
\label{conditionbunded}
\end{equation}
where $f_{R_{j^*}^{[1]}}(r)$ is the probability density function of $R_{j^*}^{[1]}$. Next, given that $R_{1}=\frac{X_1}{p}$, we know that $f_{R_{1}}(r)=pf_{X_1}(pr)$. Accordingly, we have
\begin{equation}
f_{R_{1}}(r)\leq pM.
\end{equation}
Given that $f_{R_{1}}(r)$ is bounded, the same holds then for $f_{R_{1}^{[1]}}(r)$. With this in mind, we can leverage \cite[Theorem~1]{schatte1984asymptotic} to conclude the existence of two constants $\alpha<1$ and $C>0$ such that
\begin{equation}
\sup_{r\in[0,1]} |f_{R_j^{[1]}}(r)-1|\leq C\alpha^{j},
\label{resultslmhem}
\end{equation}
In particular, these constants are equal to
\begin{equation}
\alpha=\sup_{n\in\mathbb{N}^*}|\hat{g}(n)|,
\end{equation}
\begin{equation}
C=\frac{\sum_{n=1}^{\infty}|\hat{g}(n)|^{2}}{\alpha^2},
\end{equation}
where
\begin{equation}
\hat{g}(n)=\int_{0}^{\infty}f_{R_1}(r)e^{-2\pi inr}dr.
\end{equation}
Lastly, by using the Fourier scaling property and the fact that $f_{X^{[p]}_j}(x)=\frac{1}{p}f_{R_j^{[1]}}(\frac{x}{p})$, we can conclude our proof.
\section{Proof of Theorem \ref{theoremfinalnewprocess}}
\label{appendixtheoremfinalnewprocess}
Let us start our proof by recalling the definition of the average utility loss in any interval $[0,T]$. Particularly, we have
\begin{equation}
\overline{\mathcal{L}}_T=\frac{1}{T}\int_{0}^{T}U(t)W(t)dt.
\label{expectedlossfinitetimemiddleee}
\end{equation}
Then, given eq. (\ref{rewritingfunctionBt}), we can rewrite $\overline{\mathcal{L}}_T$ as follows
\begin{equation}
\overline{\mathcal{L}}_T=\underbrace{\frac{1}{T}\int_{0}^{T}U'(t)W(t)dt}_{\overline{\mathcal{L}}'_T}+\underbrace{\frac{1}{T}\int_{0}^{T}B(t)W(t)dt}_{\overline{\mathcal{L}}''_T}.
\end{equation}
With the above in mind, we recall that our goal is to show that for every $\epsilon>0$, we have
\begin{equation}
\Pr(|\overline{\mathcal{L}}_T-\overline{\mathcal{L}}_{\infty}|\geq\epsilon)\leq g(T,\epsilon),
\end{equation}
where $g(T,\epsilon)\xrightarrow{T\rightarrow \infty}0$ for any positive $\epsilon$. To prove this, we leverage the triangular inequality and the union bound to conclude that
\begin{equation}
\Pr(|\overline{\mathcal{L}}_T-\overline{\mathcal{L}}_{\infty}|\geq\epsilon)\leq \Pr(|\overline{\mathcal{L}'}_T-\overline{\mathcal{L}}_{\infty}|\geq\frac{\epsilon}{2})+\Pr(|\overline{\mathcal{L}}''_T|\geq\frac{\epsilon}{2}).
\label{dividingintotwo}
\end{equation}
Given that $U'(t)$ is a periodic function, we can conclude from Theorem \ref{theoremfinal} that there exists a function $g_1(T,\epsilon)$ such that
\begin{equation}
\Pr(|\overline{\mathcal{L}'}_T-\overline{\mathcal{L}}_{\infty}|\geq\frac{\epsilon}{2})\leq g_1(T,\epsilon),
\end{equation}
and $g_1(T,\epsilon)\xrightarrow{T\rightarrow \infty}0$ for any positive $\epsilon$. What remains is to show similar results for the second term of eq. (\ref{dividingintotwo}). To do so, we first note that
\begin{equation}
    \mathbb{E}[\overline{\mathcal{L}}''_T]=\frac{1}{T}\int_{0}^{T}\mathbb{E}[B(t)W(t)]\overset{(a)}=\frac{1}{T}\int_{0}^{T}\mathbb{E}[B(t)]\mathbb{E}[W(t)]=0,
\end{equation}
where $(a)$ is due to the independence between $B(t)$ and $W(t)$. Next, we investigate the variance of the term $\overline{\mathcal{L}}''_T$. To that end, given that $\mathbb{E}[\overline{\mathcal{L}}''_T]=0$, we have $\mathrm{Var}[\overline{\mathcal{L}}''_T]=\mathbb{E}[(\overline{\mathcal{L}}''_T)^2]$. Consequently, we derive the second order moment of $\overline{\mathcal{L}}''_T$ below:
\begin{align}
    \mathbb{E}[(\overline{\mathcal{L}}''_T)^2]&=\frac{1}{T^2}\mathbb{E}[(\int_{0}^{T}B(t)W(t)dt)(\int_{0}^{T}B(s)W(s)ds)]\nonumber\\&
   \overset{(a)} =\frac{1}{T^2}\mathbb{E}[\int_{0}^{T}\int_{0}^{T}B(t)W(t)B(s)W(s)dtds]\nonumber\\&
   \overset{(b)} =\frac{1}{T^2}\int_{0}^{T}\int_{0}^{T}\mathbb{E}[B(t)B(s)]\mathbb{E}[W(t)W(s)]dtds
   \nonumber\\&
   \overset{(c)}\leq\frac{1}{T^2}\int_{0}^{T}\int_{0}^{T}\mathbb{E}[B(t)B(s)]dtds]\nonumber\\&
   \overset{(d)} =\frac{1}{T^2}\int_{0}^{T}\int_{0}^{T}\mathrm{Cov}[B(t),B(s)]dtds
\end{align}
where $(a)$ is the result of iterated integrals, $(b)$ follows from the independence between $B(t)$ and $W(t)$, $(c)$ springs from the fact that $0\leq W(\cdot)\leq1$, and $(d)$ is due to $B(\cdot)$ being of zero mean. Now, using the wide sense stationary property along with a change of integration variable, we can rewrite $\mathbb{E}[(\overline{\mathcal{L}}''_T)^2]$ as follows
\begin{equation}
    \mathbb{E}[(\overline{\mathcal{L}}''_T)^2]=\frac{1}{T^2}\int_{0}^{T}\int_{-s}^{T-s}\rho(\tau)dsd\tau.
\end{equation}
An illustration of the integration domain is reported in Fig. \ref{integrationregion}. Given the symmetry of the domain $\mathcal{D}$ with respect to the horizontal axis, along with the symmetry of $\rho(\cdot)$, we can manipulate the integration variables to obtain the following:
\begin{equation}
    \mathbb{E}[(\overline{\mathcal{L}}''_T)^2]=\frac{2}{T^2}\int_{0}^{T}(T-\tau)\rho(\tau)d\tau. 
\end{equation}
Then, by introducing the absolute value, we get
\begin{equation}
    \mathbb{E}[(\overline{\mathcal{L}}''_T)^2]\leq\frac{2}{T^2}\int_{0}^{T}|T-\tau||\rho(\tau)|d\tau. 
\end{equation}
Now, given the fact that $0\leq\tau\leq T$, we can deduce that
\begin{equation}
    \mathbb{E}[(\overline{\mathcal{L}}''_T)^2]\leq \frac{2}{T}\int_{0}^{T}|\rho(\tau)|d\tau.
\end{equation}
With the above in mind, and given Assumption \ref{assumptionaboutbt} and Chebyshev's inequality, we can conclude that there exists a function $ g_2(T,\epsilon)$ such that
\begin{equation}
\Pr(|\overline{\mathcal{L}''}_T|\geq\frac{\epsilon}{2})\leq g_2(T,\epsilon),
\end{equation}
and $g_2(T,\epsilon)\xrightarrow{T\rightarrow \infty}0$ for any positive $\epsilon$. This concludes our proof. 
\begin{figure}[!ht]
\centering
\includegraphics[width=.99\linewidth]{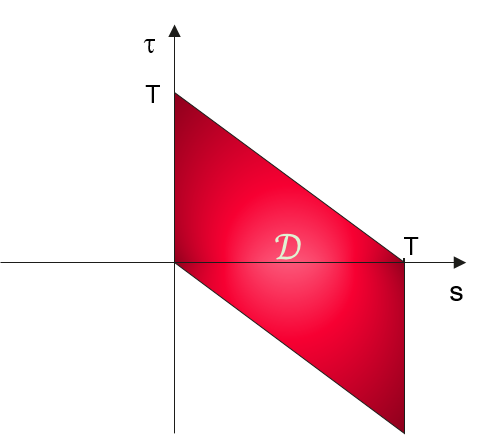}
\vspace{-20pt}
\caption{Illustration of the integration region related to $\mathbb{E}[(\overline{\mathcal{L}}''_T)^2]$.}
\label{integrationregion}
\end{figure}
\end{document}